\documentclass[lettersize,journal]{IEEEtran}
\usepackage{amsmath,amsfonts}
\usepackage{array}
\usepackage[caption=false,font=normalsize,labelfont=sf,textfont=sf]{subfig}
\usepackage{textcomp}
\usepackage{stfloats}
\usepackage{url}
\usepackage{verbatim}
\usepackage{graphicx}
\usepackage{cite}
\usepackage{xcolor}
\usepackage{caption} % 调整caption居中
\usepackage{placeins} % 强制reference在图后展示
\usepackage{float}
\usepackage{afterpage}
\usepackage{amsmath,amssymb,amsfonts}
\usepackage[ruled,vlined,linesnumbered]{algorithm2e}
\usepackage{booktabs}
\usepackage{array}
\usepackage{multirow}
\usepackage{enumitem}
\usepackage{makecell}  
\usepackage{bbm}
\usepackage{orcidlink}
\usepackage{amsthm}
\newtheorem{assumption}{Assumption}
\newtheorem{theorem}{Theorem}
\newcommand{\E}{\mathbb{E}}
\newcommand{\J}{\mathcal{J}}

\begin{document}
\title{Toward Blockage-Resilient 6G-V2X Connectivity: Semi-Distributed Bandit with Dynamic Arm Set for mmWave HetNets}
%\title{Semi-Distributed Blockage-Aware Bandit for Vehicular User Association in mmWave HetNets for 6G and Beyond}
\author{Weiqi Chi, \IEEEmembership{Graduate Student Member, IEEE,}
Bo Qian, \IEEEmembership{Member, IEEE,}
Hanlin Wu, 
Donghui Li, \\ 
Haibo Zhou, \IEEEmembership{Fellow, IEEE,}
and Manabu Tsukada, \IEEEmembership{Member, IEEE}
%\thanks{Manuscript received XXX; revised XXX; accepted XXX.}%
\thanks{This work was supported in part by the National Natural Science Foundation of China under Grant 62271244, in part by JSPS Grant-in-Aid for Early-Career Scientists under Grant JP25K21195, in part by JST ASPIRE under Grant JPMJAP2325. A preliminary conference version of this work is available as a preprint at arXiv:2606.08118~\cite{chi2026band}.}
\thanks{Weiqi Chi, Bo Qian \emph{(Corresponding author)}, Hanlin Wu, Donghui Li, and Manabu Tsukada are with the Graduate School of Information Science and Technology, The University of Tokyo, Tokyo 113-8657, Japan (e-mail: weiqichi@g.ecc.u-tokyo.ac.jp; boqian@ieee.org; hanlinwu@g.ecc.u-tokyo.ac.jp; li-donghui@g.ecc.u-tokyo.ac.jp; mtsukada@g.ecc.u-tokyo.ac.jp).}
\thanks{Haibo Zhou is with the School of Electronic Science and Engineering, Nanjing University, Nanjing 210023, China (e-mail: haibozhou@nju.edu.cn).}
}
% The paper headers
\markboth{IEEE Transactions on Vehicular Technology}%
{Chi \MakeLowercase{\textit{et al.}}: Semi-Distributed Blockage-Aware Bandit for User Association in mmWave Vehicular HetNets}
\maketitle

\begin{abstract}
The vision for 6G vehicle-to-everything (V2X) communications demands reliable, adaptive connectivity for fully autonomous driving across complex dynamic environments. Millimeter-wave (mmWave) user association (UA) in heterogeneous vehicular networks presents a particularly demanding instance of this problem, where dynamic blockages and rapid channel variations continuously undermine the stationary reward assumptions of traditional multi-armed bandit (MAB) frameworks. This paper proposes a fully distributed blockage-aware non-stationary dynamic bandit algorithm (BAND) and its semi-distributed extension S-BAND for cooperative learning across vehicles. Blockage prediction is incorporated into the change-detection (CD) mechanism to suppress false alarms, while a dynamic base station (BS) set management scheme balances exploration and exploitation across large-scale BS deployments without requiring centralized channel state information (CSI) acquisition or offline training. In S-BAND, vehicles accumulate BS reward estimates as local knowledge and periodically upload them to the macro base station (MBS), which aggregates them into cluster-based central knowledge. A trajectory-aligned knowledge (TAK) region is proposed to capture the spatial correlation of mmWave channel characteristics. A knowledge inheritance fidelity (KIF) metric is introduced to quantify knowledge transfer quality. Simulation results on a realistic urban topology show that BAND and S-BAND achieve \textbf{34.9\%} and \textbf{59.4\%} regret reduction relative to a centralized MAB baseline, with performance gains sustained across blockage rates ranging from 10\% to 50\%. The proposed TAK region consistently outperforms the traditional K-means clustering scheme under both fidelity criteria.
\end{abstract}


\section{Introduction}

Intelligent Transportation Systems (ITS) sit at the core of the smart city vision, with vehicle-to-everything (V2X) communications serving as a central enabler by allowing connected autonomous vehicles (CAVs) to exchange real-time information with surrounding vehicles, infrastructure, and network nodes. Supported applications span safety-critical collision avoidance to bandwidth-intensive cooperative perception, each imposing distinct constraints on latency and throughput~\cite{survey2}. Safety messages demand near-perfect reliability under tight latency budgets, while cooperative perception generates data volumes that narrowband technologies such as Dedicated Short-Range Communications are unable to accommodate~\cite{survey1,QB2}. The 6G mobile communication is expected to accelerate these demands further, envisioning AI-native management as one of the foundation enablers for next-generation autonomous driving~\cite{survey6G}, with anticipated ubiquitous coverage raising the performance floor that communication infrastructure must reliably meet.

Millimeter-wave (mmWave) technology offers abundant spectrum in the 30-to-300~GHz range, enabling the multi-gigabit data rates that modern vehicular applications require. IEEE 802.11bd and NR V2X have been standardized to address the throughput and latency limitations of earlier technologies~\cite{survey3}. The 6G roadmap anticipates further extension toward sub-THz bands with AI-native physical layer designs~\cite{survey6G}, yet mmWave remains the most mature tier of the emerging heterogeneous access stack. However, the propagation characteristics of mmWave introduce distinct challenges, as free-space path loss is substantially higher than at sub-6~GHz and signal penetration through buildings and vehicles is severely limited, leaving links highly susceptible to static obstructions and moving blockages. High vehicle speeds compound to this effect, as brief shadowing events can trigger abrupt connectivity loss requiring handover to an alternative base station (BS)~\cite{Li2022mobility}. The combination of mobility-induced channel non-stationarity and dynamic blockages introduces inherent intermittency in the mmWave vehicular link~\cite{survey5}, imposing demanding requirements on any network management strategy aimed at sustaining reliable connectivity.

User association (UA) determines which BS a mobile node connects to at any given time. In mmWave vehicular heterogeneous networks (HetNets), mmWave sub-BS (SBS) are densely deployed alongside conventional macro BS (MBS) to extend spectral capacity~\cite{ua1}. The UA decision directly shapes load distribution, spectrum utilization, and the Quality of Service (QoS) experienced across the HetNets. The vehicular environment substantially increases the difficulty of this problem. Dynamic blockages generated by surrounding traffic introduce abrupt link quality degradations, and large vehicles in particular produce prolonged shadowing that continuously destabilizes active association decisions under congested conditions. The high mobility of vehicles further shortens the validity window of any association decision, giving rise to a highly dynamic and unstable network topology. Conventional UA strategies are poorly matched to this regime. Methods that rely on channel state information (CSI) face the practical difficulty that measurements become outdated before association decisions can be executed~\cite{QB1}, while the cost of continuous re-estimation is prohibitive in a network of this scale and dynamism~\cite{uasurvey2016}. Max-SINR association~\cite{ua2} avoids explicit channel estimation but systematically concentrates traffic on high-power macro BSs. This leaves mmWave small cells underutilized and failing to deliver the spectral gains intended in heterogeneous deployments. Given the need for low-overhead and real-time decisions, multi-armed bandit (MAB) frameworks have attracted interest as an online alternative, where association quality is learned through sequential interaction rather than from pre-collected channel measurements or offline training~\cite{rw_mab1}. The fundamental difficulty is that standard MAB formulations assume stationary reward distributions, a condition that vehicular environments structurally violate due to mobility, blockages, and time-varying interference. This paper addresses the resulting gap, and the specific contributions are as follows.
\begin{enumerate}
\item \textbf{We propose a semi-distributed blockage-aware non-stationary dynamic bandit (S-BAND) framework for mmWave vehicular UA.} S-BAND introduces a semi-distributed architecture in which vehicles execute a local UA algorithm and periodically upload accumulated knowledge to the MBS for central knowledge formation. A knowledge transfer process is designed to reduce cold-start latency for newly arrived vehicles without requiring centralized CSI collection or offline training. S-BAND also provides a blockage-aware change-detection (CD) mechanism and a dynamic BS set management scheme that together handle the non-stationary environment and large arm set inherent to dense mmWave deployments.

\item \textbf{We design a novel trajectory-aligned knowledge (TAK) region for cooperative knowledge transfer.} Unlike conventional K-means regions, the TAK region aligns cluster-based knowledge boundaries with the road direction, better matching the spatial anisotropy of urban mmWave channels. This yields sharper geographic resolution, lower inter-cluster ambiguity, and higher knowledge transfer fidelity for entering vehicles.

\item \textbf{We propose the knowledge inheritance fidelity (KIF) metrics to evaluate cooperative knowledge transfer quality.} KIF quantifies how accurately inherited knowledge reflects a vehicle's actual channel conditions. Unlike throughput-based metrics, KIF provides a principled diagnostic for comparing knowledge region geometries and synchronization intervals independently of overall communication rate.

\item \textbf{We provide the first regret analysis for a CD-based non-stationary bandit under vehicular blockage.} We introduce a blockage-filtered baseline assumption that prevents blockage-induced observations from collapsing the false-alarm bound. We further replace the Bernoulli reward assumption of~\cite{liu2018change} with Hoeffding-based bounds for general bounded rewards, better reflecting 
continuous achievable-rate observations. The proposed algorithm achieves $\mathcal{O}(|\bar{\mathcal{J}}|\,\Upsilon_T \ln T / \Delta_{\min})$ expected regret, matching existing CD-based guarantees while covering the blockage regime.
\end{enumerate}

\section{Related Works}
\subsection{Optimization and Learning-Based Approaches to UAs}
UA represents a fundamental challenge in HetNets, becoming increasingly critical as network density, user mobility, and service diversity continue to grow. In its most general form, where each user connects to exactly one BS at any given time, the UA problem leads to complex integer nonlinear programming formulations that are typically NP-hard~\cite{Mlika2018}. Classical approaches to UA have relied on optimization-based methods such as convex optimization~\cite{rw_tradition1, rw_tradition2}, game theory~\cite{rw_tradition3}, matching theory~\cite{rw_tradition4}, and load balancing techniques~\cite{rw_tradition5}. These methods are designed to maximize network utility, balance loads across cells, and ensure QoS. However, these approaches typically rely on centralized CSI collection, which incurs prohibitive computational complexity and hinders real-time deployment in large-scale HetNets. Moreover, most methods have assumed stationary channel conditions and static optimization solutions, fundamentally limiting their adaptability to the continuous mobility and time-varying dynamics inherent in practical HetNet deployments.

Notably, reinforcement learning (RL) has emerged as a powerful approach for UA problem in mobile networks and HetNets. These problems are inherently non-convex, combinatorial, and highly dynamic due to fluctuating user demands and network topologies. Various RL approaches have been deployed to address the dynamic nature of vehicular networks, including deep Q-learning variants such as DDPG~\cite{ref54} and DDQN~\cite{ref61} for handover overhead reduction and radio link failure avoidance in mmWave scenarios, Actor-Critic algorithms~\cite{ref55} for maximizing average vehicular data rates while maintaining QoS targets, and value-based methods like SARSA~\cite{ref60} for optimal cell selection during high-mobility handovers. While these model-free RL techniques demonstrate performance improvements over traditional handover schemes, they typically require extensive offline training and struggle to adapt efficiently to rapidly changing vehicular environments. This limitation motivates the implementation of online learning frameworks, particularly MAB, which enable real-time decision-making without requiring pre-trained models or complete environmental knowledge, making them especially suitable for the time-varying channel conditions and mobility patterns of vehicular networks.

\subsection{Multi-Armed Bandit for Dynamic UA}
MAB frameworks have become popular tools for tackling UA problems in wireless networks, showing their strength in addressing the exploration-exploitation trade-off. Classical MAB algorithms such as upper bound confidence (UCB) and Thompson Sampling have been applied to network selection and resource allocation due to their simplicity and online adaptability~\cite{rw_mab1}. However, these traditional MAB algorithms assume stationary reward distributions, where the expected reward of each arm remains constant over time. In vehicular mmWave networks, the environment is non-stationary due to vehicle mobility, channel variations, and dynamic blockage. To address this issue, contextual MAB (CMAB) approaches exploit side information that correlates with reward changes. % inherent in dynamic environments where network conditions, user mobility, and channel states are time-varying
Sun et al. propose a MAB-based handover mechanism that leverages received signal strength vectors and UE movement history to reduce unnecessary handovers in ultra-dense mmWave networks~\cite{2021mmWave}. Singh et al.~\cite{NS3} formulate UA as a restless bandit problem and select BS by minimum Whittle index. However, computing Whittle indices incurs significant computational overhead as the number of UEs and base stations grows, limiting its scalability in large-scale networks. He et al. exploit inter-context reward correlations in~\cite{He1} and map vehicle context into a reproducing kernel Hilbert space to capture nonlinear context-reward relationships~\cite{He2}. While richer context information enables more accurate reward estimation, it simultaneously enlarges the context space and makes the context-reward mapping harder to learn, slowing convergence as the agent requires extensive exploration across the expanded space before reliable estimation can be achieved.
%While contextual bandits can address non-stationary rewards by partitioning the context space into hypercubes, high-dimensional context spaces lead to sparse sample distributions within each hypercube. This curse of dimensionality slows learning convergence, as insufficient samples per hypercube limit reward estimation ability without extensive exploration.  
%This algorithm focuses on load balancing through implicit state-based context rather than explicit spatial-temporal features.

Another approach to handle non-stationarity is to explicitly account for changes in the reward distribution over time, which can be broadly categorized into passively and actively adaptive policies. Passively adaptive policies are unaware of changes in the reward distribution but update their decisions based on recent observations to track the best arms. Representative algorithms include Discounted UCB~\cite{dcucb2006}, Sliding-Window UCB~\cite{swucb2008}, and Exp3~\cite{exp3}. Actively adaptive policies implement a change-detection (CD) algorithm to monitor the arms' performance and restart the bandit algorithm once a breakpoint is detected. Hartland et al.~\cite{hartland} adopted the Page-Hinkley test (PHT)~\cite{PHT} to detect breakpoints in UCB policy. Liu et al.~\cite{liu2018change} developed CD-UCB framework using Two-sided cumulative sum (CUSUM) and PHT for change detection. However, unlike persistent breakpoints, dynamic blockage is transient and vehicle-dependent, occurring within very short time periods. A naive implementation of CD algorithms could trigger false alarms, mistakenly identifying temporary blockage as actual breakpoints. This leads to frequent unnecessary restarts of the bandit algorithm, preventing convergence and degrading performance.

\subsection{Fully Distributed and Semi-Distributed UA Learning}
A natural extension of single-agent UA learning is the fully distributed multi-agent setting, where each vehicle operates as an independent agent and learns from local observations without coordination. Sana et al.~\cite{Sana2020} demonstrate that this architecture achieves substantial network-wide throughput gains in dynamic mmWave networks while eliminating inter-agent signaling overhead. The cost is that independently learning agents cannot reuse experience accumulated elsewhere in the network, leaving newly arrived vehicles to learn from scratch. Semi-distributed architectures introduce a central coordination node to address this gap. In the bandit domain, Leng et al.~\cite{Leng2022} show that selective inter-agent information sharing can similarly reduce the impact of biased observations on learning quality. These coordination mechanisms, however, rely on offline pre-training, making them ill-suited to the non-stationary reward distributions induced by vehicular mobility and dynamic blockage.

\subsection{Notation}
We represent scalars and sets by italic (e.g.\ $x$, $X$) and calligraphic letters (e.g.\ $\mathcal{X}$), with $|\mathcal{X}|$ denoting cardinality. Vectors and matrices are represented by lowercase and uppercase boldface letters (e.g.\ $\mathbf{x}$, $\mathbf{X}$). The superscripts $(\cdot)^{*}$ and $(\cdot)^\top$ denote the optimal selection and the transpose operator, $\bar{(\cdot)}$ the empirical mean, and $\mathcal{O}(\cdot)$ the computational complexity.

\section{System Model}
\label{System Model}
We investigate HetNet structure in mmWave vehicular communication where a sub-6 GHz MBS overlays multiple mmWave SBSs. The MBS functions as both an access point and a center unit for learning tasks, while SBSs enable advanced mmWave V2X applications. The rest of this section presents the vehicle mobility model and the channel model. Followed by the optimization problem formulation. 

\subsection{Mobility and Channel Model}
\label{subsec:mobilityandchannelmodel}
Consider a finite time horizon $T \in \mathbb{N}$ with discrete time steps $t = 1, \ldots, T$. At each time step, vehicles arrive in the network following a Poisson distribution. Vehicles update their states at each time period. We consider a set of SBSs denoted as $\mathcal{J}  \triangleq \{1, \ldots, j, \ldots, |\mathcal{J}|\}$, where $|\mathcal{J}|$ represents the total number of SBSs deployed for mmWave vehicular network operation. For each time step, let $\mathcal{V}^t \triangleq \{1, \ldots, k, \ldots, {|\mathcal{V}^t|}\}$ represent the vehicle set with $|\mathcal{V}^t|$ vehicles. Each SBS is equipped with a massive antenna array with $M_j$ elements, enabling serving multiple vehicles. Each vehicle $k$ is equipped with antenna arrays for both sub-6 GHz and mmWave frequency bands. For the sub-6 GHz frequency band, we adopt the widely used Gaussian MIMO channel model. For the mmWave frequency band, where one vehicle is constrained to associate with only one SBS, the channel gain between the SBS $j$ and the vehicle $k$ is given as:
\begin{equation}
y_{k,j}(t) = \mathbf{H}_{k,j}(t)\mathbf{w}_{k,j},
\label{eq.1}
\end{equation}
where $\mathbf{H}_{k,j}^t$ and $\mathbf{w}_{k,j}$ denote the channel matrix and beamforming vector from vehicle $k$ to SBS $j$, respectively. We focus on large-scale fading characteristics, including path loss and shadowing induced by buildings in the urban scenario. Since orthogonal frequency allocation is adopted across SBSs, inter-cell interference is negligible, and only intra-cell interference needs to be considered. Specifically, when vehicle $i$ simultaneously communicates with SBS $j$ at time $t$, it introduces interference to the ongoing transmission toward vehicle $k$, characterized by the interference channel coefficient $\tilde{y}_{i,j}(t)$. The total interference at SBS $j$ while serving vehicle $k$ is therefore aggregated over all vehicles that simultaneously communicate with SBS $j$:
\begin{equation}
I_{k,j}(t) = \sum_{i \in (\mathcal{V}^t \setminus k)} P_{v}\left|\tilde{y}_{i,j}(t)\right|^2\mathcal{I}_{i,j}^t  + N_o W,
\label{eq.3}
\end{equation}
where the $W$ is the bandwidth of SBS $j$, $P_{v}$ is the transmission power of vehicles, $N_o$ is the noise power density and $\mathcal{I}_{i,j}(t) = 1$ indicates $i$ is associating with SBS $j$, otherwise $\mathcal{I}_{i,j}(t) = 0$.

\subsection{Blockage Model}
\label{subsec:blockagemodel}

Dynamic vehicular blockage is detected using the geometry-based method from our previous work~\cite{vtcchi}, where blockages are identified from the positional information of SBSs and vehicles. Detected blockages are treated as complete obstructions. A blocker vehicle obstructs the Line-of-Sight (LOS) link when it lies on the LOS path between an SBS and a target vehicle, and when $40\%$ of the first Fresnel zone is obstructed~\cite{fresnel}. The Fresnel zone radius at the obstruction is:
\begin{equation}
\tilde{r}=\sqrt{\lambda_c\frac{d_{so} d_{vo}}{d_{so}+d_{vo}}},
\label{eq:fresnel_radius}
\end{equation}
where $\lambda_c$ is the carrier wavelength, and $d_{so}$, $d_{vo}$ are the distances from the SBS and vehicle to the obstruction vehicle along the LOS path. Applying knife-edge diffraction theory, blockage occurs when the obstructing vehicle height $z_o$ exceeds the effective Fresnel ellipsoid height:
\begin{equation}
\tilde{z}=\left(z_v - z_{\text{s}}\right)\frac{d_{so}}{d_{so}+d_{vo}} + z_{\text{s}} - 0.6\tilde{r},
\label{eq:effective_height}
\end{equation}
where $z_{\text{s}}$ and $z_v$ are the heights of the SBS and target vehicle. Let $\mathcal{B}^t$ indicate the collection of blocked SBS index sets for vehicle set $\mathcal{V}^t$, defined as:
\begin{equation}
\mathcal{B}^t \triangleq \{\mathbf{b}_1^t, \ldots, \mathbf{b}_k^t, \ldots, \mathbf{b}_{|\mathcal{V}^t|}^t\},
\label{eq:blocked_set}
\end{equation}
where $\mathbf{b}_k^t \in \mathcal{J}$ denotes the vector of blocked SBS indices for vehicle $k$ at time $t$.

\subsection{Optimization Problem}
\label{subsec:optimizationproblem}
According to \eqref{eq.1} and \eqref{eq.3}, the instantaneous transmission rate between vehicle $k$ to SBS $j$:
\begin{equation}
R_{k,j}^{t} = W \log_2 \left(1 + \frac{P_{v} |y_{k,j}(t)|^2}{I_{k,j}(t)}\right),
\label{eq.4}
\end{equation}
At time $t$, vehicle $k$ is associated with one SBS. Let $\boldsymbol{\eta}^{t}$ denote the vector collecting the selected SBS index for each vehicle in $\mathcal{V}^t$. The association between vehicles and SBSs can be defined as:
\begin{equation}
\boldsymbol{\eta}^{t} \triangleq [\eta_{1}^{t}, \ldots,\eta_{k}^{t}, \ldots, \eta_{{|\mathcal{V}^t|}}^{t}].
\label{eq.5}
\end{equation}

The UA optimization problem aims to maximize the sum rate of the vehicular network by finding the optimal association vector:
\begin{equation}
\begin{aligned}
\max_{\boldsymbol{\eta}^t} \quad & r(\boldsymbol{\eta}^{t}) = \sum_{{k} \in \mathcal{V}^t} R_{k, \eta_{k}^t}^{t} \\
\text{s.t.} \quad & \sum_{j=1}^{|\mathcal{J}|} \mathcal{I}_{k,j}^{t} = 1, \quad \forall k \in \mathcal{V}^t.
\end{aligned}
\label{eq.6}
\end{equation}

The constraint in~\eqref{eq.6} enforces that each vehicle associates with exactly one SBS at each time step $t$, while each SBS may serve multiple vehicles simultaneously. The optimization problem presented in~\eqref{eq.6} exhibits NP-hardness characteristics attributed to the combination of non-convex nonlinear constraints and integer variables. While exhaustive search can identify optimal associations with complete CSI knowledge, it becomes impractical in highly dynamic vehicular networks due to the high computational complexity and the accuracy of CSI estimation.

\section{Problem formulation}

In this section, we present the basic concepts of the contextual bandits problem and further formulate the vehicular UA process under a change-detection bandit framework where a piecewise-stationary environment is assumed. 

\subsection{CMAB Components for UA}
In RL-based UA, each vehicle operates as an agent that interacts with the network environment over time steps. At each time step $t$, every vehicle observes its current context and makes a decision to select an SBS for association. This sequential decision-making under uncertainty aligns naturally with the CMAB framework, where an agent repeatedly chooses among a set of arms (here, the available SBSs) to observe only the reward of the selected SBSs at each round. We now formally define the key components of the CMAB-based UA process in mmWave vehicular networks.
% MAB framework~\cite{mabbook2019}
\begin{enumerate}[label=\alph*)]
    \item \textbf{\textit{Arm:}} Each arm corresponds to one of the SBSs available for association, indexed by $j \in \mathcal{J}$.
    %\item \textbf{\textit{Agent:}} At time step $t$, the set of vehicles $\mathcal{V}^t$ act as agents, each making independent association decisions.
    %\item \textbf{\textit{Context:}} The context of vehicle $k$ at $t$ is defined as its Cartesian coordinates $\mathbf{x}_{k}^{t} \in \mathcal{X} \subseteq \mathbb{R}^2$, where $\mathcal{X}$ denotes the context space.
    \item \textbf{\textit{Context:}} The context of vehicle $k$ at $t$ is defined as its Cartesian coordinates $\mathbf{x}_{k}^{t} \subseteq \mathbb{R}^2$.
    \item \textbf{\textit{Action:}} At time step $t$, vehicle $k$ selects an SBS for association based on its current context. This action is represented by the association variable $\eta_{k}^{t} \in \mathcal{J}$.
    \item \textbf{\textit{Policy:}} A policy governs how vehicles map observed contexts to actions during the learning process. 
    \item \textbf{\textit{Reward:}} When vehicle $k$ selects SBS $j$ at time step $t$, it observes an instantaneous reward $R_{k,j}^{t}$, reflecting the immediate link communication quality.
    \item \textbf{\textit{Expected reward:}} The expected reward for vehicle $k$ selecting SBS $j$ at time step $t$ is defined as the expectation of the instantaneous reward, i.e., $\bar{R}_{k,j}^{t} = \mathbb{E}[ R_{k,j}^{t}]$. 
    \item \textbf{\textit{Trials:}} $n_{k,j}^{t}$ denotes the cumulative number of times SBS $j$ has been selected up to time step $t$.
\end{enumerate}

We consider a multi-agent CMAB setting in which multiple vehicles operate as simultaneous learners. Since mmWave propagation characteristics are highly location-dependent, the vehicle's position $\mathbf{x}_{k}^{t}$ is adopted as the context. However, the movement of neighboring vehicles continuously reshapes blockage patterns, introducing temporal variations that positional information alone cannot fully reflect. Therefore, the blocked SBS indices $\mathbf{b}_{k}^{t}$ are also incorporated into the association decisions to account for the dynamic blockage conditions.

%In mmWave vehicular networks, propagation characteristics are highly location-dependent~\cite{Va2017}. Vehicles at different positions experience distinct path loss, reflection, and blockage patterns, resulting in heterogeneous reward distributions across SBSs. A vehicle's kinematic state, such as its position and velocity, therefore serves as a natural proxy for the underlying large-scale channel conditions. That said, the movement of neighboring vehicles continuously reshapes blockage patterns and interference levels, introducing temporal variations that positional information alone cannot capture. The optimal SBS selection may thus differ substantially even among vehicles following similar trajectories. To account for both effects, we adopt the CMAB framework, where the context $\mathbf{x}_{k}^{t}$ jointly encodes the vehicle's kinematic  state and the observed blockage status of neighboring vehicles, allowing association decisions to be conditioned on both the spatial and dynamic aspects of the propagation environment.

\subsection{Upper Confidence Bound Policy}

To navigate the exploration-exploitation trade-off, we adopt the UCB policy in our algorithm design. UCB policy constructs an optimistic estimate of each arm's potential reward by adding a confidence term that reflects estimation uncertainty to the expected reward. Intuitively, an arm that has been pulled infrequently carries a wider confidence interval and is thus assigned a higher UCB value, naturally incentivizing exploration of less-visited SBSs without requiring explicit exploration parameters. As estimation improves through repeated pulls, the confidence term shrinks, shifting the balance progressively toward exploitation. Specifically, the UCB value for vehicle $k$ associated with SBS $j$ at time step $t$is computed as:
\begin{equation}
    u_{k,j}^{t} = \bar{R}_{k,j}^{t-1} + \sqrt{\frac{2 \ln t}{n_{k,j}^{t-1}}}.
\label{eq.7}
\end{equation}

The second term represents the confidence bound, which decays as SBS $j$ is selected more frequently. Vehicle $k$ then associates with the SBS that maximizes this optimistic value:
\begin{equation}
    j^* = \arg\max_{j \in \mathcal{J}}\, u_{k,j}^{t}.
\label{eq.8}
\end{equation}

Algorithm~\ref{alg.ucb} presents the UCB-based UA procedure at time step $t$ under the CMAB framework.

\begin{algorithm}[h]
\caption{CMAB UA using UCB}
\label{alg.ucb}
\SetAlgoLined
\SetInd{0.5em}{0.5em}
\textbf{Input:} $\bar{R}_{k,j}^{t-1}$, $n_{k,j}^{t-1}$ for all $j \in \mathcal{J}$\\[0.5em]
    %Identify vehicle context $x_{k}^{t}$\;
    \For{$j \in \mathcal{J}$}{
        Calculate UCB value $u_{k,j}^{t}$ according to~\eqref{eq.7}\;
    }
    Select SBS $j^*$ using~\eqref{eq.8}\;
    Receive instantaneous reward $R_{k,j^*}^t$\;
    Update number of pulls: $n_{k,j^*}^{t} = n_{k,j^*}^{t-1} + 1$\;
    Update expected reward: $\bar{R}_{k,j^*}^{t} = \bigl(\bar{R}_{k,j^*}^{t-1}\cdot n_{k,j^*}^{t-1} + R_{k,j^*}^{t}\bigr) / n_{k,j^*}^{t}$\;
\textbf{Output:} $j^*$, $n_{k,j^*}^{t}$, $\bar{R}_{k,j^*}^{t}$\\
\end{algorithm}

\subsection{Change-Detection based Bandit Framework}
\label{subsec:CDframework}

Although mmWave channels exhibit short coherence times and are highly sensitive to the surrounding environment, V2I channel studies at mmWave frequencies in urban environments have demonstrated that the non-stationary fading process admits quasi-stationarity regions of finite duration, within which channel statistics remain approximately stable before changing abruptly as the propagation environment evolves~\cite{Rodriguez2023}. This empirical evidence motivates modeling the environment as piecewise-stationary~\cite{liu2018change}, where channel statistics remain constant within each stationary segment and change abruptly at unknown breakpoints, thereby enabling natural integration of CD algorithms into the bandit policy. We formalize this model through the following assumptions and introduce the two-sided CUSUM algorithm as the CD mechanism.

\begin{assumption}[Piecewise Stationarity]
\label{ass:stat}
The reward distribution of SBS $j$ to a vehicle $k$ remains stable within stationary segments and shifts abruptly at unknown breakpoints, where a breakpoint occurs at time $t$ if $\exists j \in \mathcal{J}$ such that $\bar{R}_{k,j}^t \neq \bar{R}_{k,j}^{t+1}$. A breakpoint may be induced by mobility, blockage, or a change in the interference $I_{k,j}(t)$ from other vehicles re-associating. Any two consecutive breakpoints are separated by at least $|\mathcal{J}|P$ time slots for some integer $P$.
\end{assumption}

\begin{assumption}[Detectability]
\label{ass:detect}
There exists a known parameter $e > 0$ such that, whenever the reward of any SBS $j$ shifts between two consecutive segments, the mean change is at least $3e$. $e$ also serves as the sensitivity parameter in the CUSUM statistic in~\eqref{eq:CDalg}.
\end{assumption}

For each SBS $j$, the two-sided CUSUM-CD algorithm~\cite{liu2018change} proactively tracks simultaneous upward and downward shifts in the observed reward sequence. A baseline mean $\bar{\mu}_{k,j} \triangleq \frac{1}{M}\sum_{t=t_0}^{t_0+M-1} R_{k,j}^{t}$ is estimated from the initial $M$ observations collected by vehicle $k$ from SBS $j$, with $(g_{k,j}^{t_0+M})^{\pm} = 0$. The following cumulative statistics are maintained:
\begin{equation}
(g_{k,j}^{t})^{\pm} = \max\left(0,\, (g_{k,j}^{t-1})^{\pm} \pm 
\left(R_{k,j}^t - \bar{\mu}_{k,j} - e\right)\right).
\label{eq:CDalg}
\end{equation}

A breakpoint is declared when $(g_{k,j}^{t})^{\pm} \geq \sigma$, where $\sigma$ is the detection threshold balancing sensitivity to true breakpoints against false alarms. Upon the detected breakpoint, the bandit statistics associated with SBS $j$ are fully reset, as prior reward estimates no longer reflect the current channel conditions.

\section{The Proposed Framework}
\subsection{Framework Overview}
\label{subsec:semiFrame}

The fluctuating mmWave vehicular channel quality caused by vehicle mobility and dynamic blockages violates the stationary reward assumption underlying the MAB framework. The CMAB framework partially addresses this by incorporating contextual information. However, the expanding context space degrades the efficiency of the context-reward mapping process and severely impedes convergence. To address this challenge, this paper proposes algorithms with distinct architectural designs spanning fully distributed and semi-distributed schemes within CMAB, collectively balancing the trade-offs between convergence speed and dynamic network adaptability. As illustrated in Fig.~\ref{fig:system}, vehicles traversing in urban environment seek mmWave connectivity from densely deployed SBSs, while the MBS operates on the sub-6 GHz band to support centralized coordination. Two UA algorithms are proposed within this network, where BAND adopts a fully distributed design, and S-BAND follows a semi-distributed architecture. Under S-BAND, each vehicle independently executes a local bandit-based UA algorithm and periodically uploads its accumulated bandit statistics as local knowledge to the MBS, which aggregates into a cluster-based central knowledge available for inheritance during knowledge transfer process. 

\begin{figure*}[htbp]
    \centering
    \includegraphics[width=2\columnwidth]{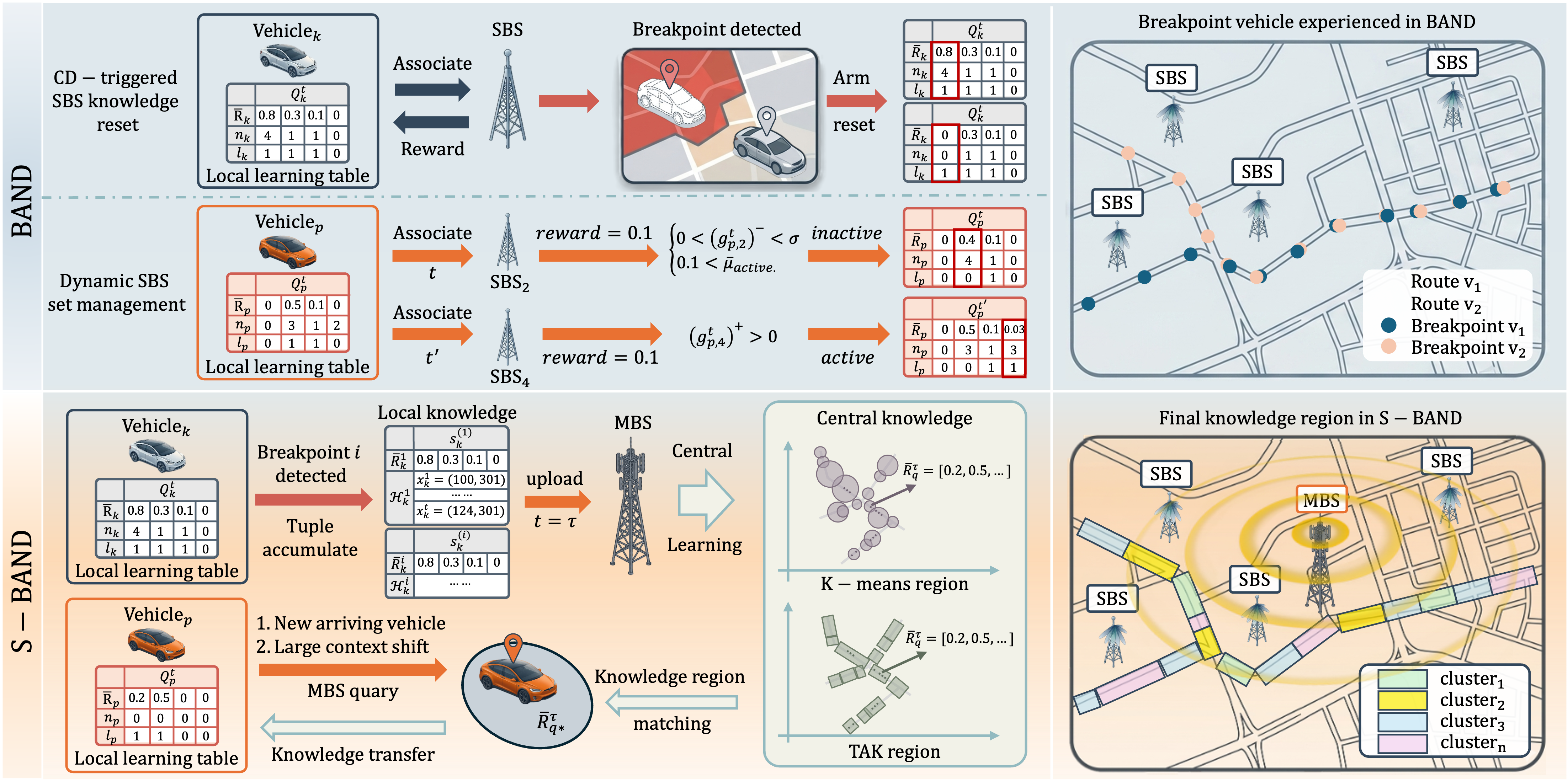}
    \captionsetup{justification=centering}
    \caption{System Overview}
    \label{fig:system}
\end{figure*}

\subsection{Blockage-Aware Non-stationary Dynamic Bandit (BAND)}
\label{subsec:BAND}
In the fully distributed BAND, each vehicle independently executes the UA decision process without any coordination, where the received reward remains subject to inter-vehicle interference as formulated in~(\ref{eq.4}). Rather than incorporating blockage status as an additional context dimension, BAND directly predicts blockage-induced reward degradation, which mitigates false-alarm restarts observed in traditional CD algorithms while relying solely on locally observable information. Beyond blockage handling, BAND introduces a two-stage UCB policy with a dynamic SBSs set management scheme, balancing exploration and exploitation across SBS sets through probability-based selection. Key components of the BAND algorithm are summarized as follows:

\subsubsection{\textbf{Local learning table}}
For each vehicle $k$ at time step $t$, the BAND algorithm maintains a local learning table $\boldsymbol{Q}_k^t \in \mathbb{R}^{3 \times |\mathcal{J}|}$ for all SBSs defined as:
\begin{equation}
\boldsymbol{Q}_k^{t} = \begin{bmatrix} (\bar{\boldsymbol{R}}_k^t)^\top ,\; (\boldsymbol{n}_k^t)^\top ,\; (\boldsymbol{l}_k^t)^\top \end{bmatrix}
\label{eq:learningtable}
\end{equation}

where $\boldsymbol{\bar{R}}_k^t = [\bar{R}_{k,j}^t]_{j=1}^{|\mathcal{J}|}$ denotes the estimated reward vector, $\boldsymbol{n}_k^t = [n_{k,j}^t]_{j=1}^{|\mathcal{J}|}$ denotes the trial vector, and $\boldsymbol{l}_k^t = [l_{k,j}^t]_{j=1}^{|\mathcal{J}|}$ denotes the BS status vector with $l_{k,j}^t \in \{0,1\}$, where $l_{k,j}^t = 1$ denotes an active SBS with accumulated observations, and $l_{k,j}^t = 0$ denotes an inactive SBS that yields low received power or has yet to be explored. The $\boldsymbol{Q}_k^{t}$ is initialized upon vehicle entry and updated whenever a new reward is received.

\subsubsection{\textbf{Two-stage UCB policy}}
Vehicle $k$ carries out the SBS association following a two-stage strategy:
\begin{itemize}
    \item \textbf{Stage 1 (SBS set selection):} With probability $\epsilon$, the vehicle explores the inactive SBS set. Otherwise, it exploits the active SBS set with probability $(1-\epsilon)$:
    \begin{equation}
    \mathcal{S}_k^t = \begin{cases}
    \{j \in \mathcal{J} \setminus \mathbf{b}_k^t : l_{k,j}^{t-1} = 1\}, & \text{w.p. } (1-\epsilon), \\
    \{j \in \mathcal{J} \setminus \mathbf{b}_k^t : l_{k,j}^{t-1} = 0\}, & \text{w.p. } \epsilon,
    \end{cases}
    \label{eq:levelselection}
    \end{equation}
    where $\mathcal{S}_k^t$ is the chosen candidate SBS set, with blocked SBSs indices $\mathbf{b}_k^t$ defined in~\eqref{eq:blocked_set} excluded due to predicted severe LOS obstruction.
    \item \textbf{Stage 2 (UCB policy):} Within the candidate set $\mathcal{S}_k^t$, the vehicle applies the UCB policy (Algorithm~\ref{alg.ucb}) restricted to SBS $j \in \mathcal{S}_k^t$. This two-stage strategy prioritizes exploitation among active SBSs that are likely to provide high-quality connections, while periodically exploring inactive SBSs to detect potential channel quality improvements.
\end{itemize}
\subsubsection{\textbf{Dynamic SBS set management}}
At the core of the BAND algorithm, the SBS set management scheme dynamically categorizes SBSs into active and inactive sets. It continuously adapts the SBS categorization to time-varying channel conditions, ensuring that the two-stage UCB policy operates on an up-to-date SBS categorization that reflects the current network environment. This management integrates the two-sided CUSUM-CD algorithm (Section~\ref{subsec:CDframework}) with a threshold-based SBS set update mechanism, and operates through three phases:

\begin{itemize}
    \item \textbf{Initialization:} When a new vehicle $k$ enters the network or experiences a significant position shift ($\Delta{\text{pos}_k} > d_{reset}$), $\boldsymbol{\bar{R}}_k^t$ and $\boldsymbol{n}_k^t$ are initialized as $\mathbf{0}$, and only SBSs whose physical distance to vehicle $k$ is within a predefined threshold $d_{init}$ are labeled as active. This distance-based initialization reduces the exploration by excluding geographically implausible candidates from the outset.

    \item \textbf{Threshold-based categorization:} The BS set status $l_{k,j}^t$ is updated based on cumulative drift $(g_{k,j}^{t})^{\pm}$ of each SBS. An active SBS $j$ is demoted to inactive if $(g_{k,j}^{t})^-$ is detected and its most recent reward falls below the mean estimated reward $\bar{\mu}_{\text{active}}$ across all currently active SBSs. Conversely, an inactive BS $j$ is promoted to active immediately if $(g_{k,j}^{t})^+$ is detected to increase the exploration probability of potentially promising SBSs.

    \item \textbf{CD-triggered reset:} When the two-sided CUSUM-CD algorithm detects a significant reward distribution shift for SBS $j$ (i.e., $(g_{k,j}^{t})^{\pm} \geq \sigma$), the knowledge of SBS $j$ is reset, enabling exploration in the new channel conditions.
\end{itemize}

These phases ensure a timely adaptation to dynamic network conditions, enabling BAND to focus on promising SBSs while avoiding low-quality associations. Notably, SBSs experiencing temporary blockage retain their previous level classification, as they are already excluded from the candidate set by the blockage-aware level selection in~\eqref{eq:levelselection}, yielding no reward observations during blockage periods. Since vehicular blockages are typically transient, this design prevents premature level re-categorization that would degrade long-term performance. The complete BAND framework is presented in Algorithm~\ref{alg:BAND}.

\begin{algorithm}[h]
\caption{BAND}
\label{alg:BAND}
\SetAlgoLined
\SetInd{0.5em}{0.5em}
\textbf{Input:} $t$, $\mathcal{V}^{t}$, $\mathcal{J}$, $\boldsymbol{Q}^0$\;
    \For{each vehicle $k \in \mathcal{V}^{t}$}{
        \If{$k \notin \mathcal{V}^{t-1}$\textbf{or} $\Delta pos_k > d_{reset}$}{
            Initialize $\boldsymbol{Q}_k^{init}$\;
        }
        Blocked SBSs indices $\mathbf{b}_k^t$ in (\ref{eq:blocked_set})\;
        SBS set $\mathcal{S}_k^t$ selection according to~\eqref{eq:levelselection}\;
        SBS selection within $\mathcal{S}_k^t$: $R_{k,j^*}^{t} \gets \text{Algorithm~\ref{alg.ucb}}$\;
        Calculate $(g_{k,j^*}^t)^{\pm}$ according to~\eqref{eq:CDalg}\;
        \scalebox{1}{$\boldsymbol{Q}_k^t \left\{\begin{array}{ll}
        \boldsymbol{Q}_{k,j^*}^{init}, & \text{if } (g_{k,j^*}^{t})^{\pm} > \sigma \\
        l_{k,j^*}^t = 0, & \text{if } (g_{k,j^*}^{t})^{-} > 0,\ R_{k,j^*}^t < \bar{\mu}_{\text{active}} \\
        l_{k,j^*}^t = 1, & \text{if } (g_{k,j^*}^{t})^{+} > 0
        \end{array}\right.$}\;
    }
\textbf{Output:} $j*$, $R_{k,j^*}^{t}$, $\boldsymbol{Q}^t_k$\;
\end{algorithm}

\subsection{Semi-distributed BAND}
\label{subsec:SBAND}

\begin{figure*}[htbp]
    \centering
    \subfloat[]{\includegraphics[width=0.24\textwidth]{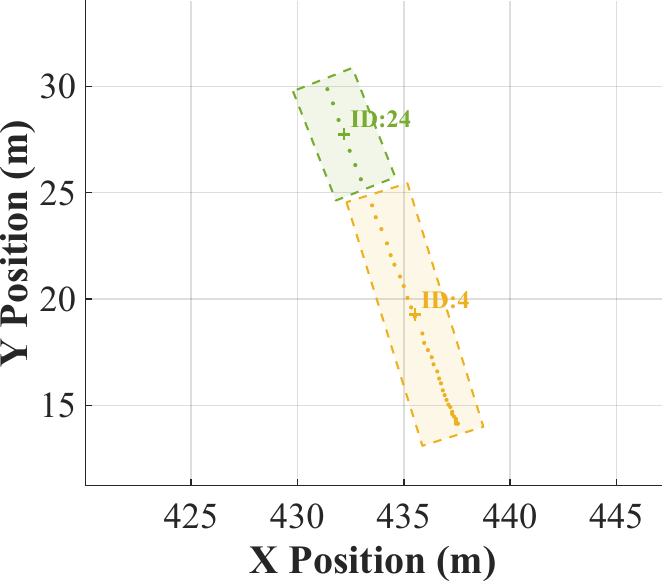}}
    \subfloat[]{\includegraphics[width=0.24\textwidth]{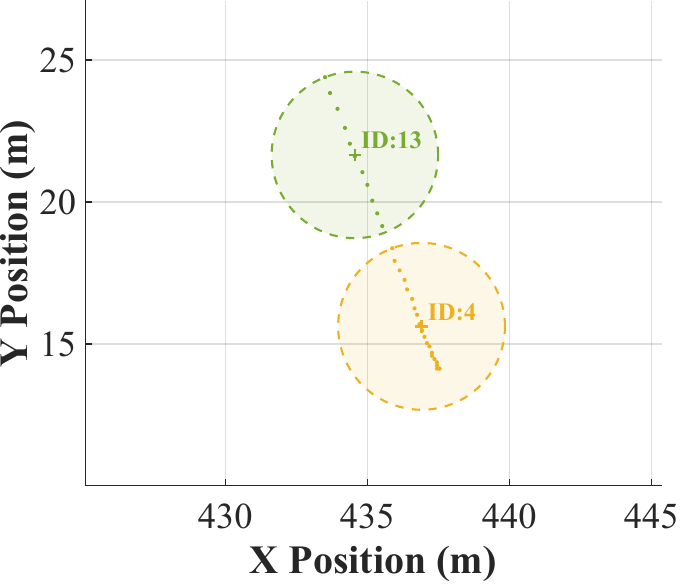}}
    \subfloat[]{\includegraphics[width=0.24\textwidth]{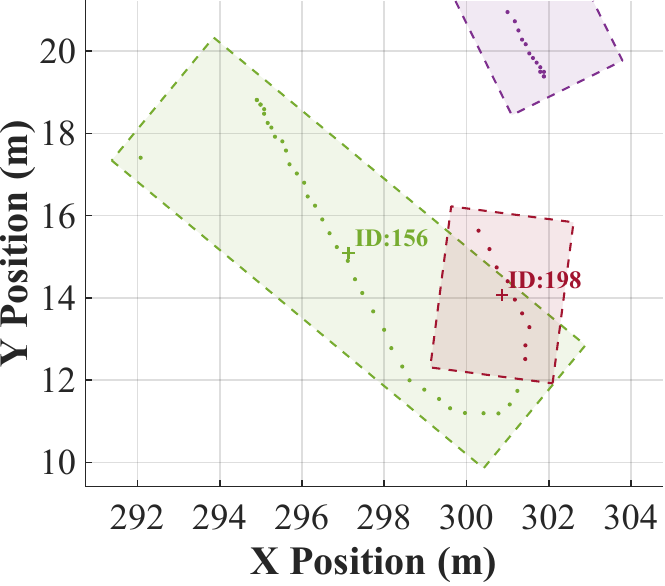}}
    \subfloat[]{\includegraphics[width=0.24\textwidth]{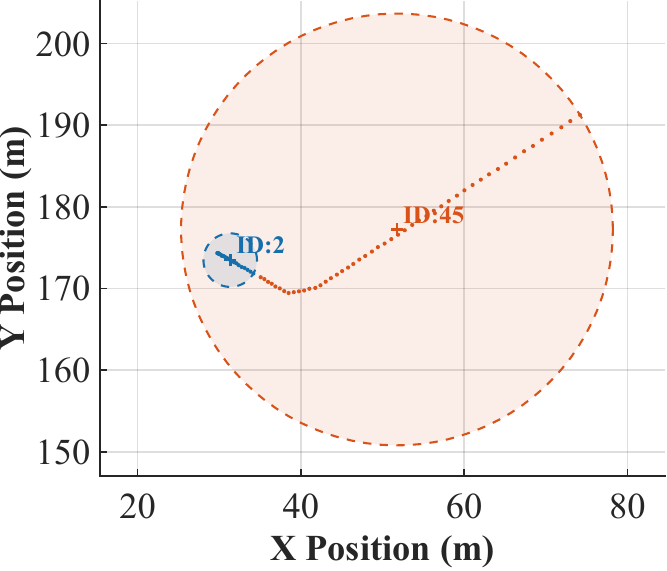}}
    
    \subfloat[]{\includegraphics[width=0.24\textwidth]{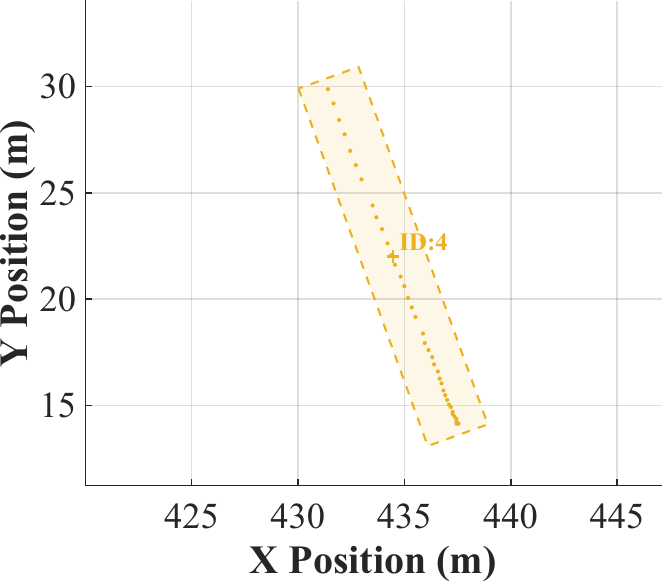}}
    \subfloat[]{\includegraphics[width=0.24\textwidth]{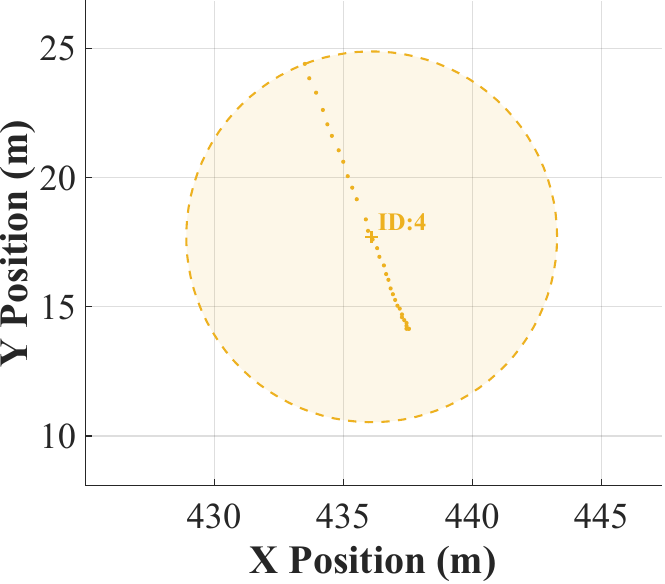}}
    \subfloat[]{\includegraphics[width=0.24\textwidth]{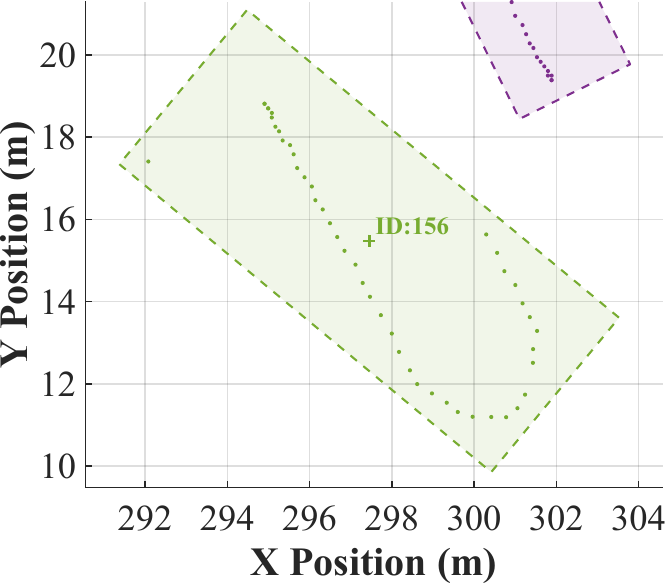}}
    \subfloat[]{\includegraphics[width=0.24\textwidth]{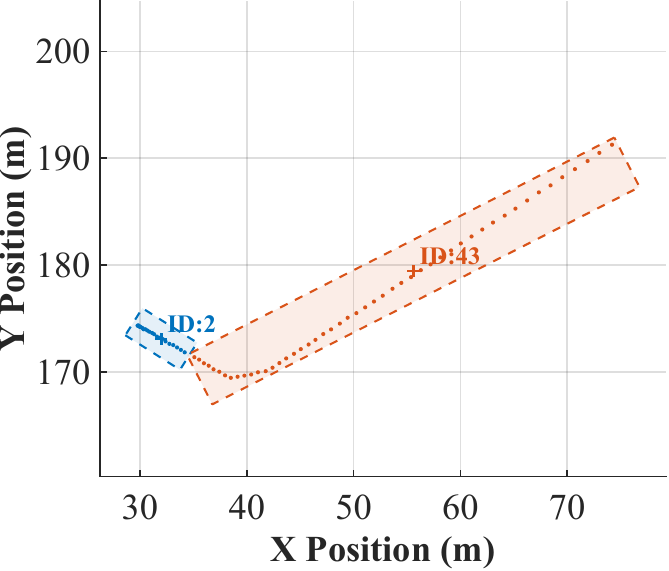}}
    
    \caption{Knowledge region merging in synchronization period with different region configurations. (a)-(e) and (b)-(f) illustrate TAK and K-means region merging, respectively; (c)-(g) shows how TAK regions handle trajectory curves; (d) and (h) demonstrate the overlapping problem that K-means regions often encounter, while TAK regions exhibit better region resolution.}
    \label{fig:context_merging}
\end{figure*}

While the fully distributed BAND algorithm offers lightweight learning with strong adaptability to non-stationary mmWave environments, each vehicle learns independently and discards all accumulated knowledge upon departure. This raises a natural question: if locally learned knowledge could be shared across vehicles, could it accelerate the learning process for newly arriving vehicles traversing similar road segments? The semi-distributed BAND (S-BAND) extension addresses this by leveraging the HetNet architecture to disseminate distributed knowledge via the MBS. The S-BAND consists of three stages.

\subsubsection{\textbf{Local learning stage}}
During the local learning stage, each vehicle $k$ maintains its local learning table $\boldsymbol{Q}_k^{t}$ and records its trajectory as a sequence of historical positions since the previous breakpoint. Upon the $i$-th breakpoint detection, the vehicle forms a tuple $\boldsymbol{s}_k^{(i)} = (\mathcal{H}_k^{(i)},\, \boldsymbol{\bar{R}}_k^{(i)})$, where $\mathcal{H}_k^{(i)}$ is the sequence of positions collected since the previous breakpoint and $\boldsymbol{\bar{R}}_k^{(i)}$ is the estimated reward vector extracted from $\boldsymbol{Q}_k^{t}$ over the same interval. The trajectory record is cleared after each tuple formation, ensuring that $\mathcal{H}_k^{(i)}$ serves as the geographical mapping of $\boldsymbol{\bar{R}}_k^{(i)}$ within the same tuple. Within each synchronization period $\tau$, multiple breakpoints may be detected, and tuples $\boldsymbol{s}_k^{(i)}$ accumulate into a tuple set as local knowledge that vehicle $k$ uploads 
to the MBS at the end of period $\tau$.

\subsubsection{\textbf{Central knowledge}}
The MBS maintains a cluster-based central knowledge $\mathcal{C}^{\tau} = \{\mathbf{c}_1^{\tau}, \ldots,\mathbf{c}_q^{\tau}, \dots, \mathbf{c}_Q^{\tau}\}$ with each cluster containing its cluster statistics:
    \begin{equation}
        \mathbf{c}_q^{\tau} = [\bar{\mathcal{A}_q^{\tau}},\, \bar{\boldsymbol{R}}_q^{\tau}]^\top,
        \label{eq:sBAND_cluster}
    \end{equation}
where $\bar{\boldsymbol{R}}_q^{\tau}$ denotes an estimated reward vector computed by averaging the individual reward vectors from over all tuples within cluster $\mathbf{c}_q^{\tau}$. $\bar{\mathcal{A}}_q^{\tau}$ represents the knowledge region of cluster $\mathbf{c}_q^{\tau}$, formed by trajectory records from all tuples in the same cluster. We calculate the knowledge region $\bar{\mathcal{A}}_q^{\tau}$ using two geometric shapes: 
\begin{itemize}
    \item \textbf{K-means region:} the minimum enclosing circle centered at the mean of all trajectory points in the cluster, with radius equal to the maximum distance from the centroid to any assigned trajectory point.
    \item \textbf{Trajectory-aligned knowledge (TAK) region:} a trajectory-aligned bounding box of minimum area enclosing the trajectory points, expanded by a fixed margin $\delta$ along each axis.
\end{itemize}

We form geographic knowledge regions, motivated by the strong spatial correlation of mmWave propagation characteristics, which makes geographic proximity a reliable surrogate for reward similarity. However, K-means region representations are ill-suited for vehicular scenarios, as road-aligned vehicle trajectories cause neighboring K-means regions to overlap excessively, degrading geographic resolution. Furthermore, building blockage in urban mmWave propagation produces spatially truncated signal patterns aligned with the road axis, which K-means regions fail to capture. Motivated by these observations, and as illustrated in Fig.~\ref{fig:context_merging} (d) and (h), we adopt a TAK region in addition to the K-means region formulation, thereby reducing inter-cluster overlap and better representing the underlying geometry of the reward distribution.

\subsubsection{\textbf{Central learning stage}}
Within each synchronization period $\tau$, the MBS updates $\mathcal{C}^{\tau}$ by merging newly uploaded tuples into $\mathcal{C}^{\tau-1}$. For each incoming tuple $\boldsymbol{s}_k^{(i)}$, the MBS determines whether it should be merged into an existing cluster $\boldsymbol{c}_q^{\tau-1}$ by checking if the tuples' knowledge region formed from $\mathcal{H}_k^{(i)}$ either spatially overlaps with the cluster region $\bar{\mathcal{A}}_q^{\tau-1}$, or if the two region centroids are within a distance threshold $d_{cen}$. If so, the MBS further verifies whether the sets of SBSs yielding non-zero rewards in $\boldsymbol{s}_k^{(i)}$ and the matched cluster are identical. If both conditions are met, the tuple is assigned to that cluster; otherwise, $\boldsymbol{s}_k^{(i)}$ is registered as a new cluster. Once all tuples are assigned, the knowledge region $\bar{\mathcal{A}}_q^{\tau}$ is recomputed from the combined trajectory positions, and the reward vector $\bar{\boldsymbol{R}}_q^{\tau}$ is updated by averaging across all assigned tuples.

\subsubsection{\textbf{Knowledge Inheritance}}
When a vehicle $k$ enters the network or its position shifts are larger than $d_{reset}$, it queries the MBS if there's any central knowledge for inheritance. The MBS evaluates spatial containment against all clusters in $\mathcal{C}^{\tau}$ based on the vehicle's current position $\mathbf{x}_k^t$: 
\begin{equation}
    \mathcal{C}_k^{\text{match}} = \left\{\boldsymbol{\mathbf{c}}_q^{\tau} \in \mathcal{C}^{\tau} : 
    \mathbf{x}_k^t \in \mathcal{A}_q^{\tau}\right\},
    \label{eq:inherit_match}
\end{equation}
If multiple clusters match, the vehicle inherits from the cluster $\boldsymbol{c}_q^{\tau}$ whose knowledge region $\bar{\mathcal{A}}_q^{\tau}$ has its geometric centroid closest to $\mathbf{x}_k^t$:

\begin{equation}
    \boldsymbol{c}_{q^*}^\tau = \arg\min_{q \,:\, \boldsymbol{c}_q^{\tau} \in \mathcal{C}_k^{\text{match}}} \left\| \bar{\boldsymbol{a}}_q^{\tau} -  \mathbf{x}_k^t \right\|_2,
    \label{cluster_optimal}
\end{equation}
where $\bar{\boldsymbol{a}}_q^{\tau}$ denotes the geometric centroid of $\bar{\mathcal{A}}_{q}^{\tau}$. The inherited reward vector $\bar{\boldsymbol{R}}_{q^*}^{\tau}$ initializes the local learning table as $\bar{\boldsymbol{R}}_k^t = \bar{\boldsymbol{R}}_{q^*}^{\tau}$, $\boldsymbol{n}_k^t = \boldsymbol{0}$, and $\boldsymbol{l}_k^t$ is initialized such that $l_{k,j}^t = 1$ for all SBSs $j$ with $\bar{R}_{q^*,j}^{\tau} > 0$, and $l_{k,j}^t = 0$ otherwise. To validate the inherited knowledge, the vehicle compares the inherited reward $\bar{R}_{q^*,j}^{\tau}$ against observed rewards during the first association steps. If $\left|\bar{R}_{q^*,j}^{\tau} - R_{k,j}^t\right| > \Delta$, the inherited knowledge is discarded and the vehicle reverts to local learning from scratch.

\begin{algorithm}[h]
\caption{S-BAND}
\label{alg:sBAND}
\SetAlgoLined
\SetInd{0.5em}{0.5em}
\textbf{Input:} $T$, $\mathcal{C}^{0}$, $\mathcal{V}^{t}$, $\mathcal{J}$, $\tau$, $\eta$\\[0.5em]
\For{$t = 1$ \KwTo $T$}{
    \If{$t \bmod \tau = 0$}{
        Vehicles upload accumulated tuples\;
        $\mathcal{C}^{\tau} \gets \text{Central learning in MBS}$\;
    }
    \For{$k \in \mathcal{V}^{t}$}{
        \If{$k \notin \mathcal{V}^{t-1}$ \textbf{or} $\Delta pos_k > d_{reset}$}{
            \tcp{Knowledge inheritance query}
            \eIf{$\mathcal{C}_k^{\text{match}} \neq \emptyset$}{
                Identify $\boldsymbol{c}_{q^*}^\tau$ according to~\eqref{cluster_optimal}\; 
            }{
                Initiate local learning table $\boldsymbol{Q}_k^{init}$\;
            }
        }
        $ R_{k,j^*}^t \gets \text{Execute Algorithm~\ref{alg:BAND}}$\; 
        \If{$\mathcal{C}_k^{\text{match}} \neq \emptyset$ \textbf{and} $\bigl|\bar{R}_{q^*,j}^{\tau} - R_{k,j^*}^t\bigr| > \Delta$}{
            Initiate $\boldsymbol{Q}_k^{init}$\;
        }
        Update local knowledge $\boldsymbol{Q}_k^{t}$\;
        \If{Breakpoint $i$ is detected for $k$}{
           $\boldsymbol{s}^{(i)}_k \gets [\mathcal{H}_k^{(i)},\, \boldsymbol{\bar{R}}_k^{(i)}]$ \tcp*{Tuple formation}}
    }
}
\textbf{Output:} Associated SBS indices $\boldsymbol{\eta}$\\
\end{algorithm}
\subsection{Regret Analysis}
\label{sec:regret}
We bound the regret of \textsc{BAND} under a piecewise-stationary 
environment, building on Assumptions~\ref{ass:stat} and~\ref{ass:detect} 
introduced in Section~\ref{subsec:CDframework}. We further impose the following assumptions tailored to the structure of \textsc{BAND} to ensure the CUSUM statistic tracks the true channel distribution rather than transient blockage.
\begin{assumption}[Blockage-Filtered Baseline]
\label{ass:blk}
The CUSUM baseline $\bar{\mu}_{k,j}$ is estimated using only samples 
collected when SBS $j$ is not blocked.
\end{assumption}

We consider a representative vehicle $k$ and analyze its \emph{per-vehicle marginal regret}, measured against the per-slot optimal SBS under the realized interference process, where the reward $R^t_{k,j}$ is the achievable rate conditioned on the contemporaneous associations of all other vehicles, i.e., the inter-vehicle interference $I_{k,j}(t)$ is part of the observed reward. We normalize rewards to $[0,1]$ and let $\Upsilon_T$ denote the number of breakpoints in $[1,T]$, where in the context of \textsc{BAND}, any change in the other vehicles' behavior that shifts the reward distribution seen by vehicle $k$ is detected as a breakpoint by the CUSUM module in the same way as a mobility- or blockage-induced shift. We do not claim a system-wide guarantee, such as a social optimum or a Nash equilibrium, over the joint association of all vehicles. The convergence of the coupled multi-agent dynamics is beyond the scope of this regret analysis and is left to future work. The empirical convergence of all vehicles under this coupling is nonetheless demonstrated in Section~\ref{result3}, Fig.~\ref{fig:all_regret}.

Our analysis uses the result of \cite{liu2018change} at two levels. 
The regret decomposition of \cite[Theorem~1]{liu2018change} requires only Assumption~\ref{ass:stat} and treats the mean detection delay $\E[D]$ and the expected number of false alarms $\E[F]$ as inputs, which we invoke directly. To bound $\E[D]$ and $\E[F]$ themselves, however, \cite[Theorem~2]{liu2018change} additionally assumes Bernoulli rewards, which does not hold here since our reward is the achievable rate $R^t_{k,j}$, bounded in $[0,1]$ but otherwise general. We therefore bound the change-detection performance via Hoeffding's inequality for bounded variables rather than the Bernoulli-specific argument of \cite{liu2018change}. Under Assumptions~\ref{ass:stat}--\ref{ass:blk}, the negative CUSUM drift required for the false-alarm bound is ensured by Assumption~\ref{ass:blk}, and there exist constants $C_1,C_2>0$ depending only on $e$ and $M$ such that
\begin{equation}
\E[D]\le C_2(\sigma+1),
\qquad
\E[F]\le 2T\,e^{-2(\sigma-Me)^2/M},
\label{eq:cd}
\end{equation}
where $\sigma$ is the detection threshold. This distinguishes our analysis from \cite{liu2018change} in two respects. We drop the Bernoulli assumption and instead handle general bounded rewards, and the blockage-filtered baseline of Assumption~\ref{ass:blk}, which is absent in \cite{liu2018change}, keeps the drift negative under transient blockage.

\begin{theorem}[Regret of BAND]
\label{thm:regret}
Under Assumptions~\ref{ass:stat}--\ref{ass:blk}, the expected cumulative regret of BAND satisfies
\begin{equation}
\E[\mathcal{R}_k(T)]\le
\bigl(\Upsilon_T+\E[F]\bigr)\frac{4|\bar{\J}|\ln T}{\Delta_{\min}}
+\Upsilon_T\,\E[D]
+\mathcal{O}\!\bigl(|\bar{\J}|\bigr),
\label{eq:regret}
\end{equation}
where $|\bar{\J}|$ is the average active-SBS-set size and $\Delta_{\min}$ the smallest reward gap over all segments.
\end{theorem}

\begin{proof}
Following the decomposition of \cite[Theorem~1]{liu2018change}, the regret has three sources. \emph{(i) Exploration.} Within each stationary phase, the UCB index \cite[Theorem~1]{auer2002finite} selects a suboptimal SBS at most $4\ln T/\Delta_{\min}^2$ times, and the number of phases equals the CUSUM resets $\Upsilon_T+\E[F]$, giving the first term. \emph{(ii) Delay.} Each of the $\Upsilon_T$ breakpoints incurs at most $\E[D]$ slots of unit regret, giving the second term. \emph{(iii) Residual.} Per-phase constants and the $M$ samples used to re-estimate each baseline give $\mathcal{O}(|\bar{\mathcal{J}}|)$. Since exploration runs only over the active set, $|\bar{\mathcal{J}}|$ replaces $|\mathcal{J}|$. 
\end{proof}

Assumption~\ref{ass:blk} is essential. Without it, a blocked fraction $p_b$ of near-zero samples shifts the CUSUM drift from $-e$ to $-e+p_b\bar{R}_{k,j}$. Once this exceeds $e$, the drift turns positive, the false-alarm bound \eqref{eq:cd} collapses, and the first term of \eqref{eq:regret} grows uncontrolled. Choosing $\sigma=Me+\sqrt{(M/2)\ln(T/\Upsilon_T)}$ yields $\E[F]=\mathcal{O}(\Upsilon_T)$ and $\E[\mathcal{R}_k(T)]=\mathcal{O}(|\bar{\mathcal{J}}|\,\Upsilon_T\ln T/\Delta_{\min})$, matching the order of \cite{liu2018change} and approaching the $\Omega(\sqrt{T})$ lower bound of \cite{swucb2008} up to log factors, with $|\bar{\mathcal{J}}|$ replacing $|\mathcal{J}|$.

\section{Simulation Result}
\label{Result}

We evaluated the performance of the proposed algorithms in a realistic scenario. The simulation environment is an urban mmWave vehicular network scenario spanning an area of $550 \times 540$ m. The urban road topology and building infrastructure are derived from OpenStreetMap~\cite{osm2017} data for the Shibuya district in Tokyo, Japan. To ensure realistic evaluation conditions, the simulation incorporates a dense mmWave BS deployment distributed throughout the region, combined with authentic vehicular traffic patterns generated using SUMO~\cite{SUMO2018}. Notably, the BS placement partially adopts real-world deployments sourced from OpenCelliD~\cite{opencellid}, supplemented by artificially placed BSs along major roads to account for future Roadside Unit deployments. The simulation map and BS locations are illustrated in Fig.~\ref{fig:BS_location}. The composition of simulated vehicle types follows the specifications in 3GPP TR 37.885, while the proportion of trucks is adjusted to emulate different dynamic blockage rates. Channel characteristics are modeled using the Clustered Delay Line (CDL) model in conjunction with ray tracing in 3GPP TR 38.901, enabling accurate representation of signal propagation affected by static building obstructions. The complete network and traffic simulation parameters are summarized in TABLE~\ref{tab:sim_params}, which serves as the default configuration for the results presented in~\ref{Result}.
\begin{figure}[htbp]
    \centering
    \includegraphics[width=0.9\columnwidth]{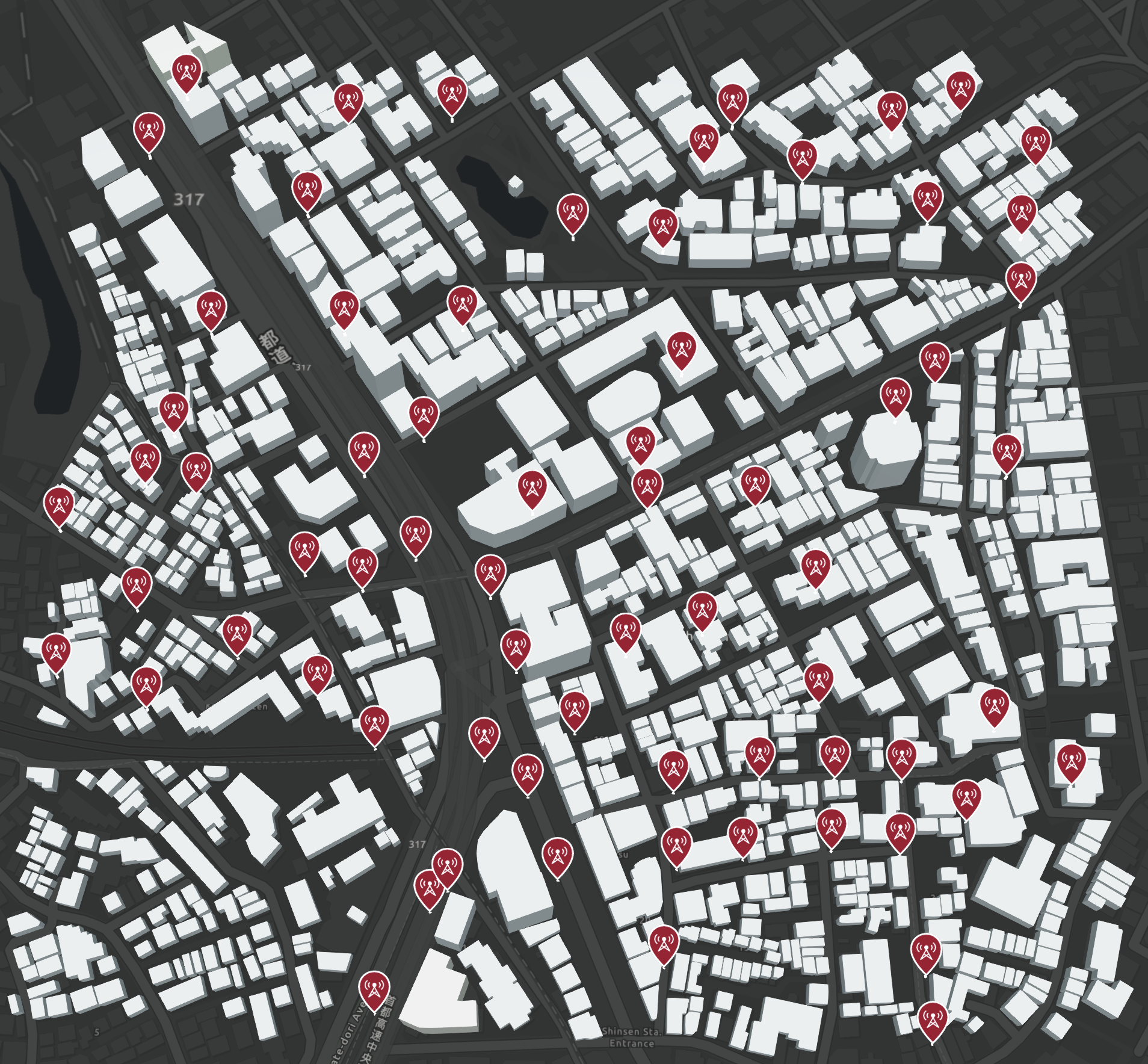}
    \captionsetup{justification=centering}
    \caption{Simulation environment: road topology and BS deployment in the Shibuya district, Tokyo.}
    \label{fig:BS_location}
\end{figure}

\subsection{Change Detection Threshold Analysis}
\label{result1}
The proposed \texttt{BAND} algorithm resets the reward estimation of an SBS upon detected breakpoints, governed by the sensitivity and detection thresholds $(\sigma, e)$ as defined in (\ref{eq:CDalg}). To investigate the impact of $(\sigma, e)$ on algorithm performance, we conduct a grid search over $\sigma \in \{0.1, 0.3, 0.5, 0.7, 0.9\}$ and $e \in \{0.05, 0.15, 0.25, 0.35, 0.45\}$, yielding 25 hyperparameter combinations, the results of which are presented in Fig.~\ref{fig:level_heat}.

\begin{table}[h]
\centering
\caption{Simulation Parameters and Settings}
\label{tab:sim_params}
\begin{tabular}{@{}lcc@{}}
\toprule
\textbf{Category} & \textbf{Parameter} & \textbf{Value} \\
\midrule
\multirow{7}{*}{\makecell{Scenario\\Setup}}
& Number of BSs & 69 \\
& Height of BSs & 5 m \\
& Simulation area & $550 \times 540$ m \\
& Blockage rate & $30\%$ \\
& \multirow{3}{*}{\makecell{Vehicle\\Dimensions}} & $5 \times 2 \times 0.75$ m \\
& & $5 \times 2 \times 1.6$ m \\
& & $13 \times 2.6 \times 3$ m \\
\midrule
\multirow{6}{*}{\makecell{Channel\\Parameters}}
& Carrier frequency & 28 GHz \\
& Bandwidth & 50 MHz \\
& Transmit power & 30 dBm \\
& BS antenna size & $4 \times 4$ \\
& Vehicle antenna size & $2 \times 2$ \\
& Noise power spectral density & $-174$ dBm/Hz \\
\bottomrule
\end{tabular}
\end{table}

\begin{figure}[htbp]
    \centering
    \subfloat[Average Communication Rate]{\includegraphics[width=0.49\columnwidth]{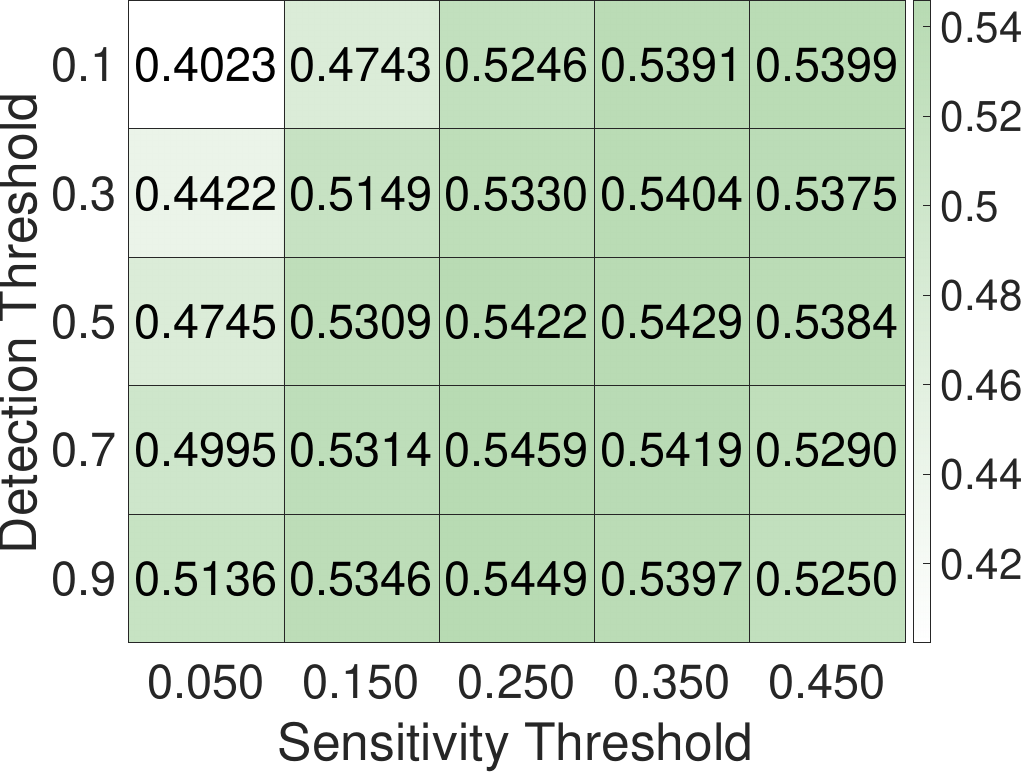}}
    \hfill
    \subfloat[Cumulative Regret]{\includegraphics[width=0.49\columnwidth]{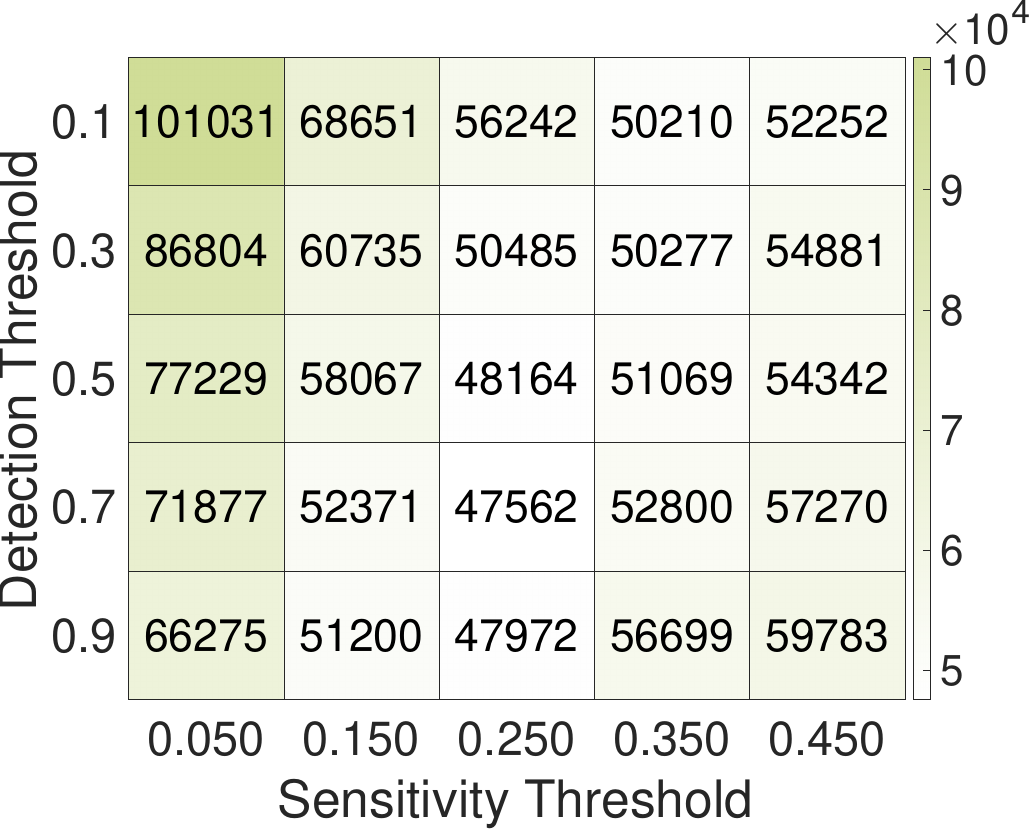}}
    \captionsetup{justification=centering}
    \caption{BAND performance with variation of $\sigma$ and $e$.}
    \label{fig:level_heat}
\end{figure}

As shown in Fig.~\ref{fig:level_heat}, the network average communication rate shows an optimal threshold selection region with $\sigma \in [0.5, 0.7]$ and $e \in [0.25, 0.35]$. The accumulative regret gives the same pattern and shows a trend that performance degrades noticeably at small $\sigma$ and small $e$ combinations, particularly at $(\sigma, e) = (0.1, 0.05)$ where the highest regret of approximately $1.01 \times 10^5$ is recorded. This joint observation suggests that a sensitivity margin of $e \in [0.25, 0.35]$ tolerates moderate signal fluctuations before accumulating evidence toward a breakpoint, effectively filtering transient channel noise while remaining responsive to genuine distributional shifts. A choice of $\sigma$ that is approximately twice the sensitivity threshold $e$ further prevents premature breakpoint declarations by requiring more accumulated evidence before triggering a reset, thereby avoiding false alarms that would degrade the learning efficiency.

\subsection{Knowledge Inheritance Fidelity analysis for S-BAND}
\label{result2}
\texttt{S-BAND} extends \texttt{BAND} with periodic centralized coordination at the MBS executed at synchronization intervals $\tau$, which governs a trade-off between premature aggregation of the local knowledge and staleness of the central knowledge. The influence of the two types of knowledge regions is also investigated. As established in~\ref{result1}, we adopt the best-performing hyperparameter combination $(\sigma, e) = (0.7, 0.25)$ for \texttt{S-BAND}. To assess the quality of knowledge inheritance, we introduce the Knowledge Inheritance Fidelity (KIF) metric, evaluated under two criteria. An inheritance is considered correct under each criterion:
\begin{itemize}
    \item \textbf{Action Fidelity:} the set of top-$3$ SBSs ranked by estimated rewards in the inherited central knowledge coincides with the set of top-$3$ SBSs ranked by the communication rates obtained from the ray-tracing simulation at the time of inheritance.
    \item \textbf{Estimation Fidelity:} the SBS yielding the highest estimated reward in the inherited knowledge coincides with the SBS providing the highest communication rate in the ray-tracing simulation at the time of inheritance, with an estimation error within $\Delta$.
\end{itemize}
Action Fidelity reflects how well the central knowledge characterizes the overall SBS behaviors within a certain region, whereas Estimation Fidelity measures the accuracy of the central knowledge in predicting the performance of the best SBS. The fidelity under each criterion is computed as the percentage of correct inheritance:
\begin{equation}
    \text{Fidelity} = \frac{{\text{Number of Correct Inheritance}}}{{\text{Number of Total Inheritance}}} \times 100\%.
\end{equation}

\begin{figure}[htbp]
    \centering
    \includegraphics[width=0.95\columnwidth]{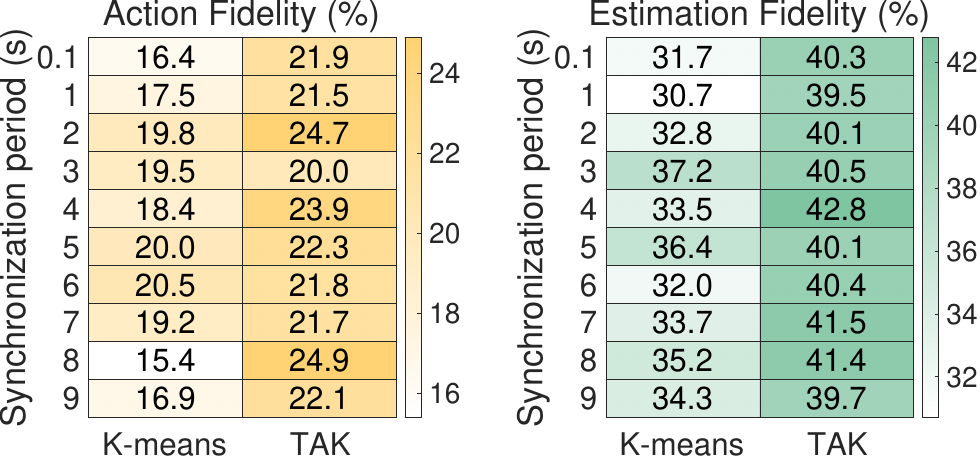}
    \captionsetup{justification=centering}
    \caption{KIF under different knowledge regions and $\tau$.}
    \label{fig:sBAND_fidelity_heat}
\end{figure}

Fig.~\ref{fig:sBAND_fidelity_heat} summarizes the KIF under both criteria across knowledge region types and synchronization intervals $\tau$. Notably, Estimation Fidelity consistently exceeds Action Fidelity, suggesting that while the central knowledge captures the performance of the top-performing SBS with reasonable accuracy, its ability to characterize the broader set of SBS behaviors within the inheritance region remains limited. Moreover, the TAK region outperforms the K-means region under both criteria. This confirms that its trajectory-aligned geometry better captures the spatial correlation of mmWave channel conditions, reducing overlap between adjacent knowledge regions and mitigating ambiguity in inheritance assignment. With respect to $\tau$, both criteria peak around $\tau = 3$--$5\,\text{s}$, reflecting the optimal balance between knowledge freshness and sufficient reward accumulation per synchronization cycle.
%The relatively low Action Fidelity values indicate that relying solely on geographic proximity to determine inheritance eligibility leaves room for improvement in matching inherited knowledge to the vehicle's actual environment. 

\begin{figure}[htbp]
    \centering
    \subfloat[Average Communication Rate]{\includegraphics[width=0.9\columnwidth]{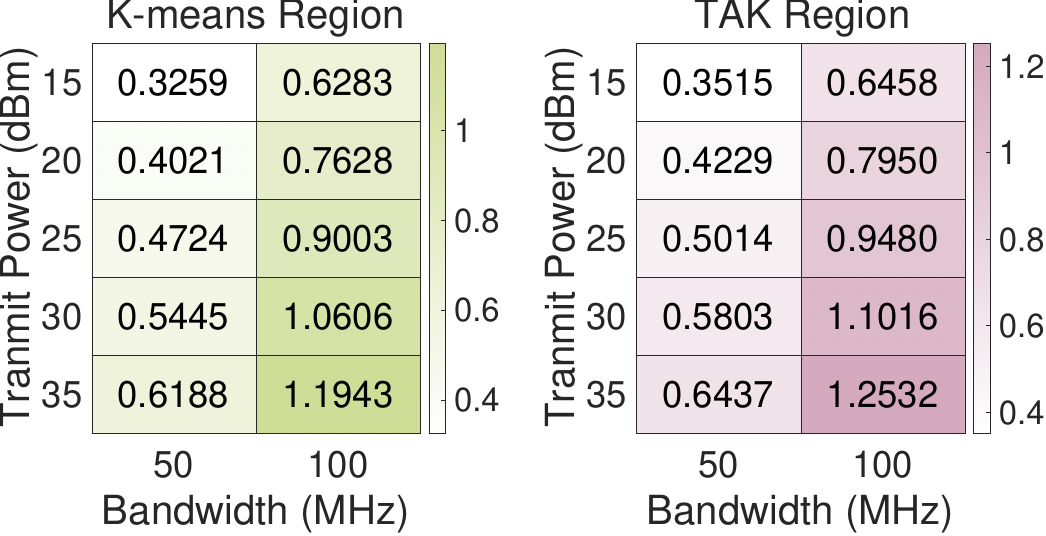}}\\
    \subfloat[Cumulative Regret]{\includegraphics[width=0.9\columnwidth]{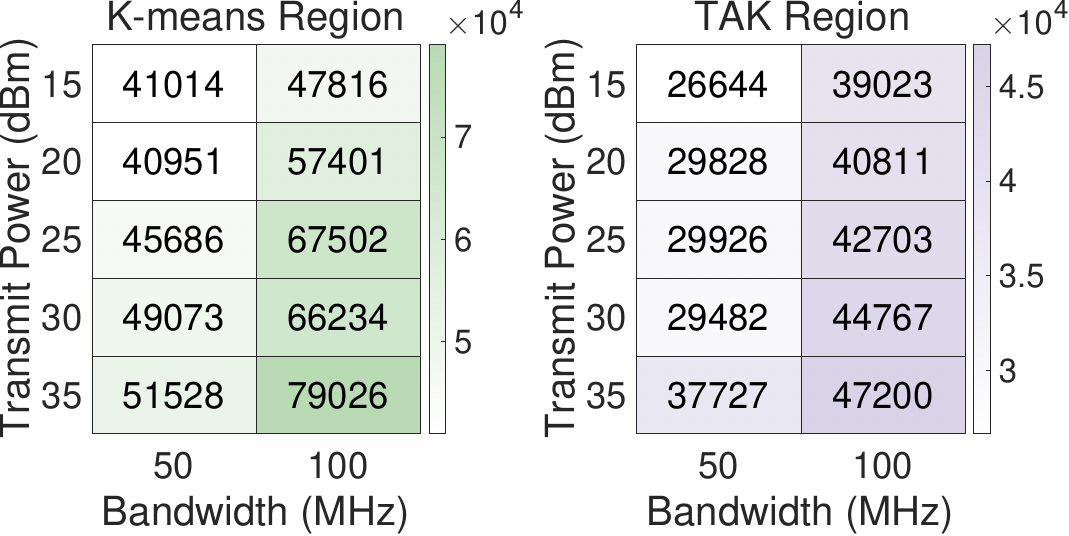}}
    \captionsetup{justification=centering}
    \caption{S-BAND performance across transmit power, bandwidth, and knowledge region at $\tau = 3\,\text{s}$.}
    \label{fig:sBAND_pwi30}
\end{figure}

Taking synchronization period $\tau = 3\,\text{s}$, Fig.~\ref{fig:sBAND_pwi30} evaluates the communication rate and cumulative regret of \texttt{S-BAND} across a range of transmit powers ($15$-$35\,\text{dBm}$) and channel bandwidths ($50$ and $100\,\text{MHz}$) for both region types. The results indicate that the TAK region achieves consistently higher communication rates and lower regret than the K-means region, confirming the advantage of trajectory-aligned spatial clustering for knowledge inheritance. Both region types benefit from increased transmit power and bandwidth, with diminishing returns at higher power levels.

\begin{figure}[htbp]
    \centering
    \includegraphics[width=0.75\columnwidth]{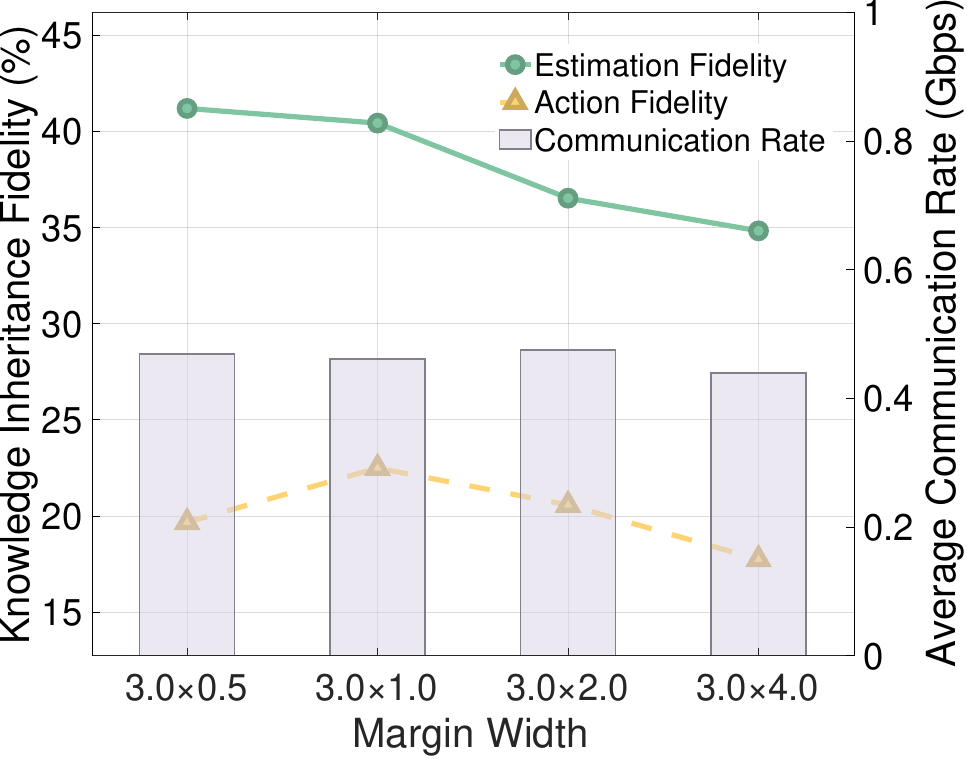}
    \captionsetup{justification=centering}
    \caption{Influence of margin $\delta$ on S-BAND performance.}
    \label{fig:sBAND_bufferL}
\end{figure}
Finally, Fig.~\ref{fig:sBAND_bufferL} investigates the effect of the fixed margin $\delta$ in TAK region construction, where a larger $\delta$ expands the bounding area of historical trajectory points, broadening the inheritance region and increasing the risk of including irrelevant channel observations. Both KIF criteria and average communication rate are evaluated with $\delta$ varied from one-half to four times the lane width of $3\,\text{m}$. Action Fidelity peaks when $\delta$ equals the lane width, and both criteria degrade as $\delta$ becomes excessively wide. The average communication rate, however, remains insensitive to $\delta$ across all tested configurations, as averaging over all vehicles obscures the localized impact of knowledge inheritance quality. 

\subsection{Learning Analysis and Blockage Effect}
\label{result3}
To comprehensively evaluate the proposed framework, the evaluated algorithms span fully centralized, distributed, and semi-distributed learning paradigms under varying blockage rates. We compare our proposed algorithms against three baseline approaches:
1) \texttt{C-UCB}: a fully centralized CMAB framework employing a UCB policy, which partitions the simulation region into geometrically defined hypercubes and assumes stationary rewards within each hypercube;
2) \texttt{minDis}: a non-learning heuristic that always associates with the nearest BS;
3) \texttt{SNR\textsubscript{thresh}}: a non-learning baseline inspired by the 3GPP A3 event-triggered handover mechanism in 3GPP TS 38.331, where a handover is initiated when the serving BS signal quality falls below a predefined threshold relative to the historical maximum received power, and the target BS is selected as the one providing the maximum reference signal received power among neighboring BSs.
The parameter configuration is summarized in Table~\ref{tab:algo_params}.

\begin{table}[h]
\centering
\caption{Algorithm Parameters}
\label{tab:algo_params}
\begin{tabular}{@{}llc@{}}
\toprule
\textbf{Algorithm} & \textbf{Parameter} & \textbf{Value} \\
\midrule
\multirow{1}{*}{C-UCB}
& Hypercube size in distance & $10 \times 10~m$ \\
\midrule
\multirow{3}{*}{BAND}
& SBS set balancing $\epsilon$ & $0.1$ \\
& CD thresholds $(\sigma, e)$ & $(0.7, 0.25)$ \\
& Distance thresholds $(d_{init}, d_{reset})$ & $(200~m, 20~m)$ \\
\midrule
\multirow{4}{*}{S-BAND}
& Accuracy tolerance $\Delta$ & $0.2$ \\
& Centroid distance threshold $d_{cen}$ & $30~m$ \\
& Mimimun Buffer width $w$ & $3~\text{m}$ \\
& Synchronization interval $\tau$ & $3~\text{s}$ \\
& Knowledge Region type & TAK region \\
%\midrule
%\multirow{1}{*}{Others}
%& UCB exploration coefficient $c$ & $1$ \\
\bottomrule
\end{tabular}
\end{table}

\begin{figure}[htbp]
    \centering
    \includegraphics[width=0.8\columnwidth]{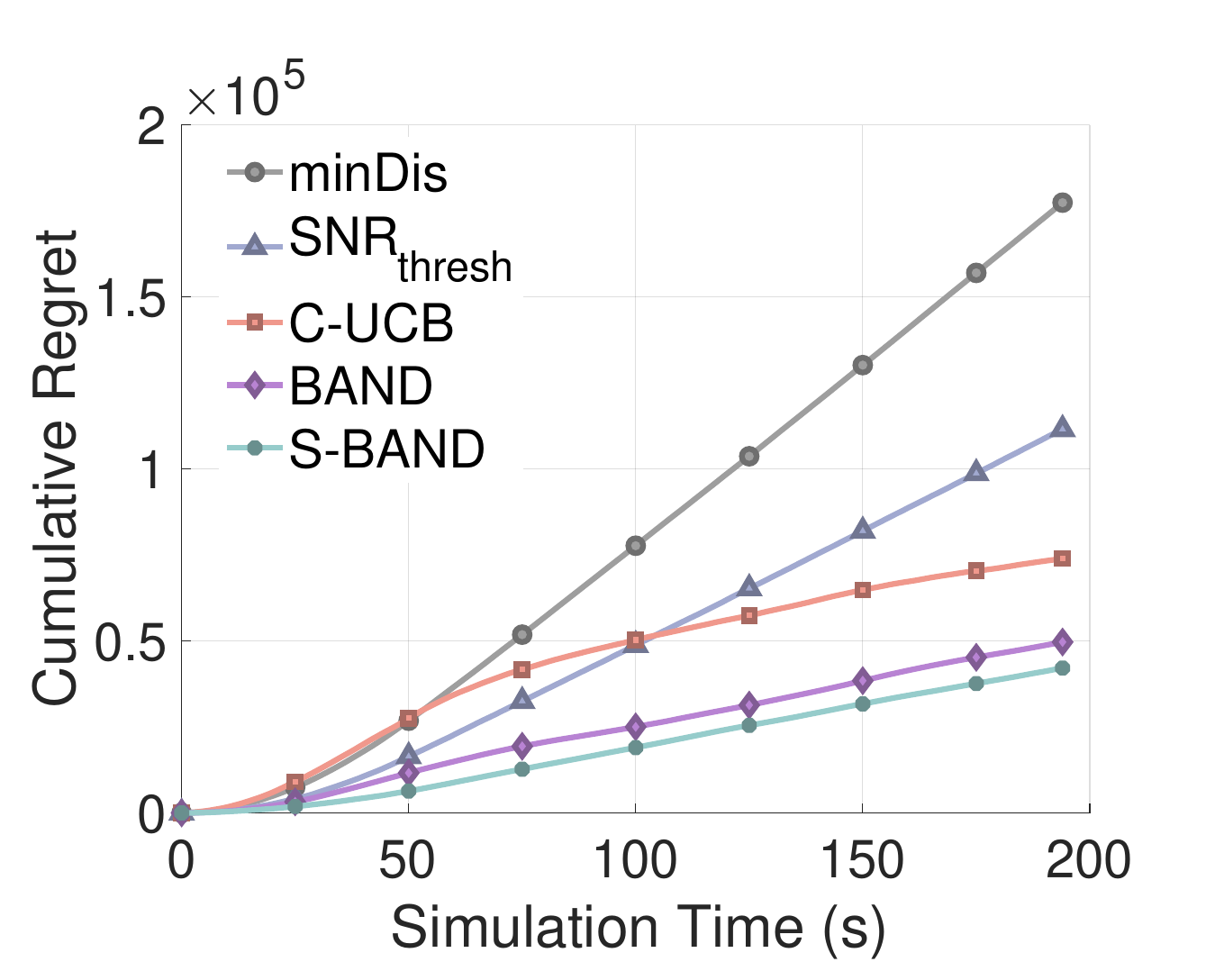}
    \captionsetup{justification=centering}
    \caption{Regret performances across all algorithms}
    \label{fig:all_regret}
\end{figure}

Fig.~\ref{fig:all_regret} compares the cumulative regret across learning time steps. The non-learning \texttt{minDis} exhibits linearly increasing regret, confirming that distance-based heuristics fail to adapt to highly dynamic mmWave channel conditions. Among the proposed bandit algorithms, \texttt{BAND} operates as a fully distributed scheme without any central coordination, yet already achieves 34.9\% regret reduction over \texttt{C-UCB}. The \texttt{S-BAND} further improves upon this, achieving 59.4\% regret reduction over \texttt{C-UCB}. All proposed algorithms converge faster than \texttt{C-UCB}, whose fixed hypercube partitioning leads to slow convergence due to sparse samples distributed across pre-defined context regions. Notably, both \texttt{BAND} and \texttt{S-BAND} operate at a coarser spatial granularity of 20\,m, yet still outperform \texttt{C-UCB}, as explicit context-reward shift tracking via local position displacement $\Delta pos_k$ compensates for the reduced partitioning resolution.

\begin{figure}[htbp]
    \centering
    \includegraphics[width=\columnwidth]{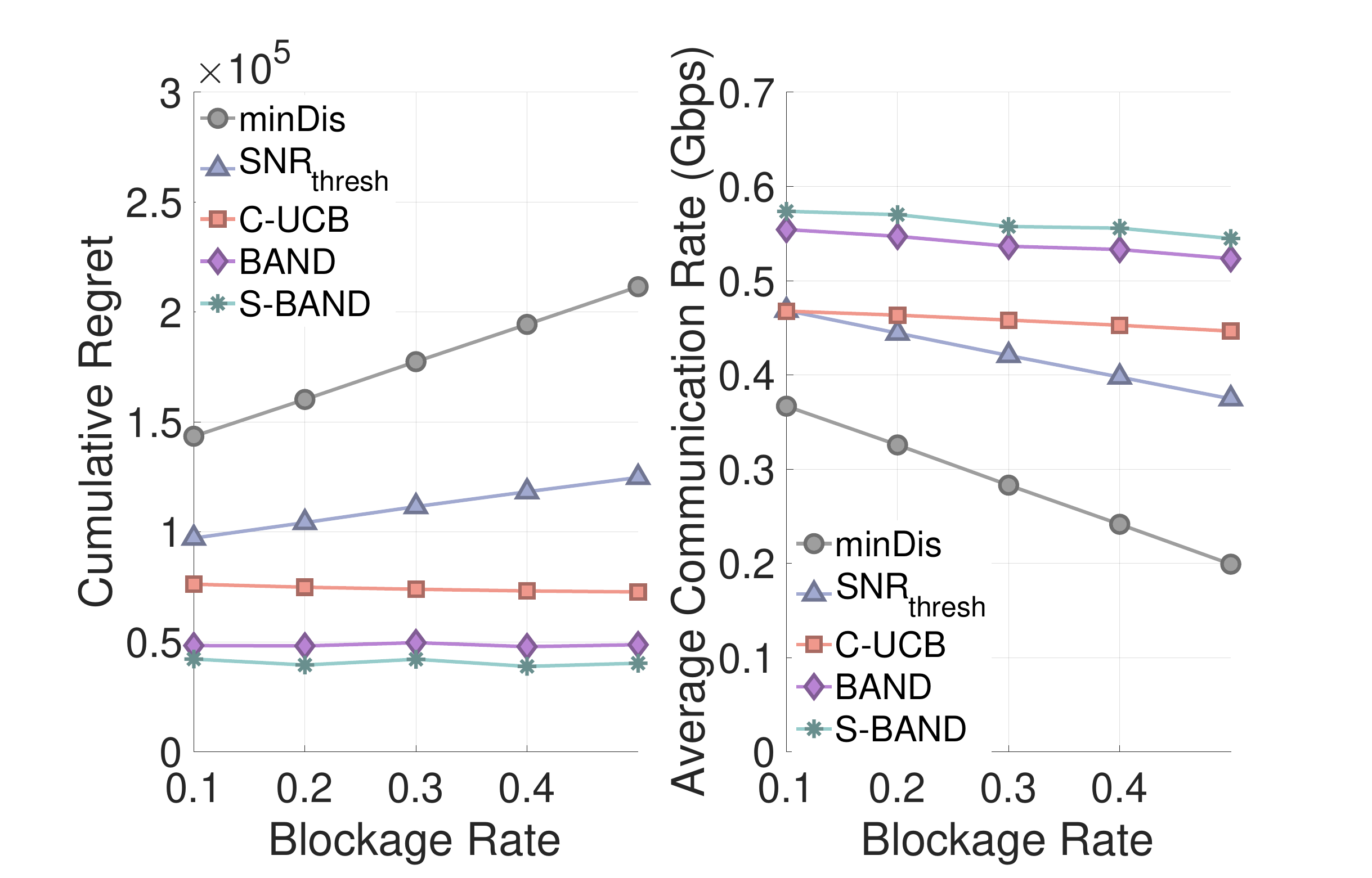}
    \captionsetup{justification=centering}
    \caption{Algorithm performance at different blockage rates}
    \label{fig:all_blockage}
\end{figure}

Fig.~\ref{fig:all_blockage} illustrates the impact of blockage rate on algorithm performance. The proposed algorithms maintain consistently lower cumulative regret and higher average communication rate across varying blockage rates from 10\% to 50\%, compared to the baseline algorithms, demonstrating the effective contribution of blockage-predictive filtering during the decision-making process. Notably, \texttt{SNR\_thresh} achieves relatively close performance to \texttt{C-UCB} when the network blockage rate is low. The consistently poor performance of \texttt{minDis} further confirms that the nearest SBS can deviate significantly from the best-performing one, particularly under dynamic blockages introduced by surrounding vehicles.

\begin{figure}[htbp]
    \centering
    \includegraphics[width=0.75\columnwidth]{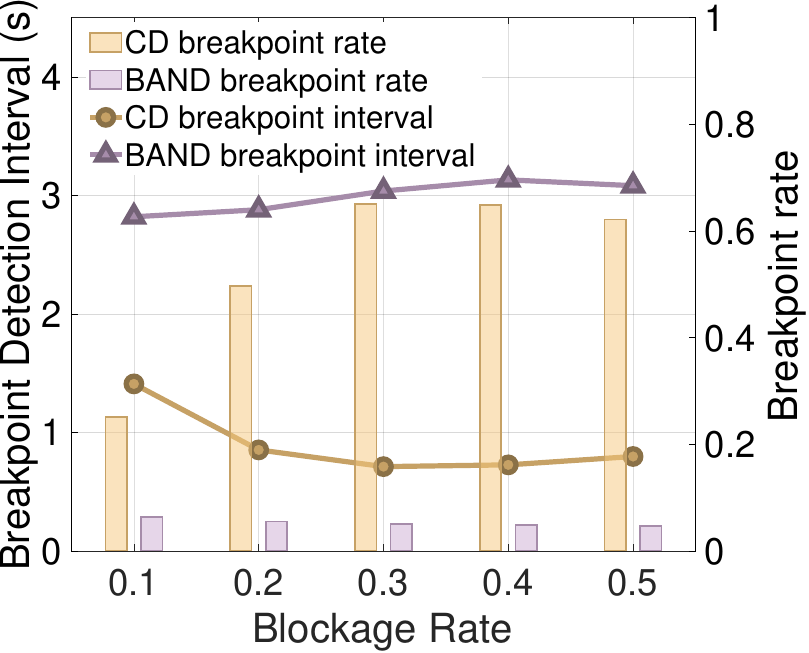}
    \captionsetup{justification=centering}
    \caption{Breakpoint detection rate and interval under different blockage rates.}
    \label{fig:breakpoint}
\end{figure}

Fig.~\ref{fig:breakpoint} evaluates breakpoint detection under varying blockage rates. Naive application of conventional CD to mmWave vehicular networks yields a substantially elevated breakpoint rate that grows with blockage intensity, as transient blockage-induced rate drops are misidentified as true reward shifts. This triggers frequent unnecessary resets, reflected in progressively shorter detection intervals that disrupt learning continuity. In contrast, \texttt{BAND} maintains a consistently low and stable breakpoint rate across all blockage conditions, demonstrating effective false alarm suppression. By explicitly distinguishing blockage-induced degradation from genuine context-reward mapping shifts, \texttt{BAND} sustains a stable reset interval, ensuring resets are effective to true reward shifts.

We give a communication overhead analysis by excluding the association execution cost, as it is identical across all approaches. \texttt{BAND} requires no central communication, as all decisions rely solely on local observations. \texttt{C-UCB} incurs one communication round per vehicle per hypercube boundary crossing, growing with vehicle speed and partitioning granularity. \texttt{S-BAND} communicates only at synchronization boundaries, incurring $\mathcal{O}(|\mathcal{V}^t| \cdot T/\tau)$ communication rounds in total, offering controllable overhead via $\tau$.

\section{Conclusion}
This paper has proposed a fully distributed and semi-distributed blockage-aware bandit framework for UA in mmWave vehicular HetNets. By explicitly differentiating transient blockage-induced degradation from genuine reward distribution shifts, the proposed blockage-aware CD mechanism has effectively suppressed false alarms under varying blockage conditions. The two-stage UCB policy with dynamic SBS set management has enabled efficient UA over densely deployed mmWave BS, requiring neither centralized CSI gathering nor offline training overhead. The S-BAND has accelerated algorithm convergence via MBS-assisted knowledge sharing, while the proposed trajectory-aligned knowledge region has better captured the geometric spatial correlation of mmWave channels, reducing inter-cluster ambiguity and improving inheritance reliability. Numerical results have demonstrated significant regret reduction over all benchmarks with robust performance sustained across varying blockage conditions. Future work will investigate adaptive synchronization intervals that dynamically adjust $\tau$ based on local channel variation rates, and joint optimization of user association and beam management to further improve spectral efficiency in dense mmWave deployments.

\FloatBarrier
\bibliographystyle{IEEEtran}
\bibliography{ref}   
\vspace{11pt}
\end{document}